\documentclass[10pt, twocolumn]{IEEEtran}
\usepackage{array,xcolor}
\usepackage{booktabs}
\usepackage{multirow}
\usepackage{comment}
\usepackage{setspace}

\renewenvironment{thebibliography}[1]{
  \begin{oldthebibliography}{#1}
    \setlength{\itemsep}{0em}
    \setlength{\parskip}{0em}
}
{
  \end{oldthebibliography}
}
\usepackage{lipsum}
\usepackage{bm}
\usepackage{tikz}
\usetikzlibrary{positioning,calc,arrows.meta,fit,backgrounds}
\usepackage{pgfplots}
\usepackage{pgfplotstable}
\usetikzlibrary{spy}
    
\usepgfplotslibrary{colormaps}
\usepackage{mathrsfs}
\usepackage{mathtools}
\usetikzlibrary{shapes.geometric, arrows}
\usepgfplotslibrary{external}
\usepackage{slashbox}

\usepackage{pifont}
\pgfplotsset{grid style={dashed,gray}}
\pgfplotsset{minor grid style={dotted,gray}}
\pgfplotsset{major grid style={dashed,gray}}

\usepackage{caption}
\usepackage{subcaption}
\usepackage{amsmath, tabularx} 
\usepackage{hyperref}
\usepackage{cleveref}
\usepackage{algpseudocode}
\usepackage{ifpdf}
\usepackage{graphicx,amssymb,lineno}
\makeatletter
\newcommand*\bigcdot{\mathpalette\bigcdot@{.5}}
\newcommand*\bigcdot@[2]{\mathbin{\vcenter{\hbox{\scalebox{#2}{$\m@th#1\bullet$}}}}}
\makeatother
\usepackage{relsize}
\usepackage{filecontents}
\usepackage{lipsum}
\usepackage{color,soul}
\usepackage{algorithm}
\usepackage{color,soul}
\usepackage{cite}

\makeatletter
\def\BState{\State\hskip-\ALG@thistlm}
\makeatother
\makeatletter
\newcommand*{\rom}[1]{\expandafter\@slowromancap\romannumeral #1@}
\makeatother
\allowdisplaybreaks

\graphicspath{ {Images/} }

\usepackage[T1]{fontenc}
\usepackage{float}

\makeatletter
\newcommand{\multiline}[1]{%
  \begin{tabularx}{\dimexpr\linewidth-\ALG@thistlm}[t]{@{}X@{}}
    #1
  \end{tabularx}
}
\makeatother

\usepackage{changepage}
\usepackage{adjustbox}
\usepackage[printonlyused]{acronym}
\newacro{6g} [6G] {Sixth generation}
\newacro{cr} [CR] {cognitive radio}
\newacro{isac} [ISAC] {integrated sensing and communication}
\newacro{sim} [SIM] {stacked intelligent metasurface}

\newacro{sb} [SB] {secondary base station}
\newacro{su} [SU] {secondary user equipment}
\newacro{sws} [SWS] {secondary wireless system}
\newacro{pb} [PB] {primary base station}
\newacro{pu} [PU] {primary user equipment}
\newacro{pws} [PWS] {primary wireless system}
\newacro{ula} [ULA] {uniform linear array}
\newacro{csi} [CSI] {channel state information}
\newacro{em} [EM] {electromagnetic}
\newacro{siso} [SISO] {single input single output}
\newacro{miso} [MISO] {multiple input single output}
\newacro{ra} [RA] {reconfigurable antenna}
\newacro{los} [LoS] {line-of-sight}
\newacro{nlos} [NLoS] {non-line-of-sight}
\newacro{shod} [SHOD] {spherical harmonious orthogonal decomposition}
\newacro{ofdm} [OFDM] {orthogonal frequency division multiplexing}
\newacro{dof} [DoF] {degree of freedom}
\newacro{fim} [FIM] {Fisher information matrix}
\newacro{bfim} [BFIM] {Bayesian Fisher information matrix}
\newacro{em} [EM] {electromagnetics}
\newacro{era} [ERA] {electromagnetically reconfigurable antenna}
\newacro{aod} [AOD] {angle-of-departure}
\newacro{aoa} [AOA] {angle-of-arrival}
\newacro{sp} [SP] {scatter point}
\newacro{ml} [ML] {maximum likelihood}
\newacro{mse} [MSE] {mean square error}
\newacro{snr} [SNR] {signal-to-noise ratio}
\newacro{rmse} [RMSE] {root mean square error}
\newacro{crb} [CRB] {Cram\'er-Rao bound}
\newacro{bcrb} [BCRB] {Bayesian Cram\'er-Rao bound}
\newacro{peb} [PEB] {position error bound}
\newacro{kld} [KLD] {Kullback-Leibler divergence}
\newacro{siso} [SISO] {single-input-single-output}
\newacro{mimo} [MIMO] {multiple-input multiple-output}
\newacro{mcrb} [MCRB] {misspecified Cram\'er-Rao bound}
\newacro{bs} [BS] {base station}
\newacro{ue} [UE] {user equipment}
\newacro{arv} [ARV] {array response vector}
\newacro{upa} [UPA] {uniform planar array}
\newacro{rf} [RF] {radio frequency}
\newacro{ris} [RIS] {reconfigurable intelligent surface}
\newacro{rms} [RMS] {root mean square}
\newacro{awgn} [AWGN] {additive white Gaussian noise}
\newacro{nsga2} [NSGA-II] {non‐dominated sorting genetic algorithm II}
\newacro{de} [DE] {differential evolution}
\newacro{iid} [i.i.d.] {independently
and identically distribute}
\newacro{sinr} [SINR] {signal to interference and noise ratio}
\newacro{sdr} [SDR] {semidefinite relaxation}
\newacro{pdf} [PDF] {probability density function}
\newacro{mpc} [MPC] {multipath component}
\newacro{lmi} [LMI] {linear matrix inequality}
\newacro{psd} [PSD] {positive semi-definite}
\newacro{rem} [REM] {radio environment map}
\newacro{it} [IT] {interference temperature}
\newacro{qos} [QoS] {quality-of-service}
\newacro{ann} [ANN] {artificial neural network}
\newacro{dnn} [DNN] {deep neural network}
\newacro{drl} [DRL] {deep reinforcement learning}
\newacro{ss} [SS] {spectrum sensing}
\newacro{dft} [DFT] {discrete Fourier transform}
\newacro{ao} [AO] {alternating optimization}
\newacro{se} [SE] {spectral efficiency}
\newacro{mmwave} [mmWave] {millimeter-wave}
\newacro{fpga} [FPGA] {field-programmable gate array}

\makeatother
\makeatother
\title{
Stacked Intelligent Metasurfaces for Multicarrier Cognitive Radio ISAC
}

\author{
Alireza Fadakar, \emph{Graduate Student Member} and Andreas F. Molisch, \emph{Fellow, IEEE}\thanks{A. Fadakar and A. F. Molisch are with the Ming Hsieh Department of Electrical and Computer Engineering, University of Southern California, Los Angeles, California, USA; E-mail: \{fadakarg,molisch\}@usc.edu.}
\thanks{This work is supported 
by the National Science Foundation (Grants 2229535 and 2106602).}
}

\usepackage{accents}

\usepackage{amsthm}

\theoremstyle{plain}

\newtheoremstyle{iremark}
  {\topsep}   
  {\topsep}   
  {\upshape}  
  {0.2in}       
  {\itshape}  
  {.}         
  {5pt plus 1pt minus 1pt} 
  {\thmname{#1}\thmnumber{ \itshape#2}\thmnote{ (#3)}} 

\usepackage{amsthm}

\newtheorem{theorem}{Theorem}
\newtheorem{lemma}[theorem]{Lemma}
\newtheorem{remark}{Remark}

\newtheorem{proposition}{Proposition}

\theoremstyle{definition}
\newtheorem*{proof}{Proof}

\makeatletter
\newcommand*\rel@kern[1]{\kern#1\dimexpr\macc@kerna}
\newcommand*\widebar[1]{%
  \begingroup
  \def\mathaccent##1##2{%
    \rel@kern{0.8}%
    \overline{\rel@kern{-0.8}\macc@nucleus\rel@kern{0.2}}%
    \rel@kern{-0.2}%
  }%
  \macc@depth\@ne
  \let\math@bgroup\@empty \let\math@egroup\macc@set@skewchar
  \mathsurround\z@ \frozen@everymath{\mathgroup\macc@group\relax}%
  \macc@set@skewchar\relax
  \let\mathaccentV\macc@nested@a
  \macc@nested@a\relax111{#1}%
  \endgroup
}
\makeatother

\begin{document}
\maketitle
\begin{abstract}
The fusion of \ac{cr} and \ac{isac}, enabled by \acp{sim}, offers a promising path for multi-functional programmable front ends in 6G and beyond. 
In this paper we propose a novel \ac{cr}-\ac{isac} framework that leverages an \ac{sim} integrated with the \ac{sb} to learn and realize optimal beampatterns that simultaneously (i) minimize the \ac{bcrb} for localizing a \ac{su} and (ii) limit averaged interference at \acp{pu} so that spectral efficiency loss is constrained, with the target of at most a few percent degradation. 
We propose an efficient alternating optimization-based algorithm to obtain the optimal end-to-end transmission response of the SIM for all \ac{ofdm} subcarriers. 
Drawing an analogy between the layered SIM architecture and deep neural networks, we define a beampattern-matching loss, derive analytical gradients for backpropagation, and implement a learning-based optimization of the SIM coefficients using a mini-batch Adam optimizer.
A complexity analysis is provided, and extensive numerical experiments are performed to evaluate the proposed \ac{cr}-\ac{isac} framework. 
The results show that the proposed \ac{sim} coefficient optimization methods attain near-optimal performance in terms of both the \ac{su} \ac{bcrb} localization metric and the \acp{pu} average spectral efficiency when the \ac{sim} has a sufficient number of layers, and they substantially outperform traditional single-layer \ac{ris} designs.
\end{abstract}
\acresetall

\bstctlcite{IEEEexample:BSTcontrol}
\section{Introduction}
Rising mobile traffic and the growing demand for higher wireless data rates have intensified radio-spectrum scarcity, posing a major challenge to the deployment of \ac{6g} and beyond communications \cite{Yang2019_6G,tataria20216,Yuan2025TTR_CR}. 
\Ac{cr} has therefore been widely adopted as an effective solution to mitigate spectrum shortage by enabling \acp{sws} (consisting of \acp{sb} and \acp{su}) to opportunistically access licensed bands, provided such access does not cause harmful interference to (typically incumbent) \acp{pws}  \cite{Liang2011CR, Xu2025CR}, \cite[Chap. 27]{molisch2026wireless}.

Spectrum sharing is typically classified into three modes: interweave, overlay, and underlay \cite{Wang2011CR_Advances, Yuan2025TTR_CR}, \cite[Chap. 27]{molisch2026wireless}. 
In the interweave mode, the \ac{sws} opportunistically exploits spectral holes that are temporally unused by the licensed \ac{pws}. 
In the overlay mode, the \ac{sws} obtains access opportunities by assisting the \ac{pws} transmission. 
In the underlay mode, the \ac{sws} may transmit concurrently with the \ac{pws} provided that the resulting performance degradation in terms of \ac{qos} at each \ac{pws} receiver remains below a predefined threshold. 
This is often achieved by requiring the  
\ac{it} to remain below a particular threshold \cite{Tanab2017Resource}. 
In this context, beamforming is essential to enhance \ac{qos} of \ac{sws} communication and sensing  while ensuring adequate protection of the \ac{pws} by limiting interference.
This paper focuses on the underlay mode due to its potential for continuous access and superior spectrum efficiency. 

On the other hand, with the rapid evolution of \ac{6g} systems, advanced enablers such as \ac{isac}, \acp{ris}, and \acp{sim}, among others, have been proposed to enhance certain channel characteristics  \cite{fadakar2025mutual, An2024SIM_DOA}. 
In particular, \ac{isac} has recently received significant attention as an efficient solution for delivering joint sensing and communication services while improving spectrum usage \cite{Nuria2024ISAC, fadakar2025mutual}. 
Furthermore, \acp{ris} and \acp{sim} can assist \ac{isac} by dynamically shaping channel amplitude and phase components to create favorable links for both communication and sensing tasks \cite{An2025sim_downlink, Niu2024SIM,fadakar2025mutual}.
Integrating these technologies with established approaches such as \ac{cr} can further improve overall \ac{6g} network performance \cite{Liu2021RISs}. 
Specifically, an \ac{ris} comprises a large number of low-cost, nearly passive reflecting elements with tunable reflection coefficients. 
Operating in a largely passive mode, \acp{ris} can be readily integrated into existing cellular and WiFi infrastructures with minimal power and signaling overhead \cite{Liu2021RISs}. 
These properties have stimulated extensive research on \ac{ris}-assisted techniques aimed at improving \ac{qos} across a wide range of communication scenarios \cite{Pan2022RIS_overview, fascista2022ris, fadakar2024multi, fadakar2025mutual, fadakar2025near}. 
However, most existing studies consider single-layer metasurfaces that provide primarily phase-only control compared to a fully digital array of the same aperture and element count, which is capable of joint amplitude and phase weighting. 
This limitation reduces the optimization degrees of freedom and prevents \acp{ris} from reproducing the full set of beampatterns achievable with fully digital arrays.

These limitations have motivated the recent investigation of \ac{sim} hardware \cite{An2023SIM_muser, An2023SIM_Holog_MIMO, Liu2025_MISO_SIM, Shi2025joint, Li2025SIM_NF, Niu2024SIM}. 
By stacking multiple programmable transmissive metasurface layers, an \ac{sim} realizes a \ac{dnn}-like architecture that offers substantially richer signal-processing capabilities than single-layer \ac{ris} implementations. 
Moreover, signal propagation within an \ac{sim} is governed by \ac{em} waves (i.e., effectively at the speed of light), whereas the execution speed of conventional \acp{dnn} is constrained by electronic processing hardware and its associated latency.
A key distinction between \acp{sim} and conventional \acp{dnn} is that \acp{sim} do not employ element-wise nonlinear activation functions, operating instead via cascaded linear electromagnetic layers \cite{An2024SIM_DOA,liu2022programmable}. 
Hardware prototypes of \acp{sim} have been demonstrated in \cite{liu2022programmable,liu2023full}.

\subsection{Related Works}
\subsubsection{RIS-aided CR Networks}
In recent years, the integration of \ac{ris} into \ac{cr} systems has been extensively studied \cite{Yuan2025TTR_CR, Dong2025joint_secure, Yuan2021IRS_CR, Zhong2022DRL_IRS_CR, Ge2024RIS_CSS_CR, Hui2024RIS_CR_HWI, Zhao2023RIS_CR, Wu2022secure_RIS_CR, Xu2025CR}.
In \cite{Yuan2025TTR_CR}, a two-timescale resource-allocation framework for \ac{ris}-aided \ac{cr} is proposed to maximize the ergodic rate of the \ac{su} subject to an average \ac{it} constraint at the \ac{pu} receiver and an average transmit-power constraint at the \ac{su} transmitter. 
The work \cite{Dong2025joint_secure} develops a joint design to strengthen secure transmission for both \acp{pu} and \acp{su} in an \ac{ris}-assisted \ac{cr} network in the presence of an eavesdropper. 
Joint optimization of the \ac{su} beamforming vector and the \ac{ris} phase-shift matrix has been investigated in \cite{Yuan2021IRS_CR} and \cite{Zhao2023RIS_CR}, where \cite{Yuan2021IRS_CR} maximizes the \ac{su} achievable rate and \cite{Zhao2023RIS_CR} minimizes the \ac{su} transmit power.
To reduce the high computational burden of conventional optimization, \cite{Zhong2022DRL_IRS_CR} proposes a \ac{drl} approach to maximize the \ac{su} rate in \ac{ris}-assisted \ac{cr} systems. 
Robust secure beamforming for \ac{mmwave} \ac{ris}-aided \ac{cr} is developed in \cite{Wu2022secure_RIS_CR}, targeting the maximization of the worst-case secrecy rate. 
From the \ac{ss} perspective, \cite{Ge2024RIS_CSS_CR} maximizes the cooperative detection probability in \ac{ris}-enabled sensing networks to improve sensing performance. 
Furthermore, \cite{Hui2024RIS_CR_HWI} addresses practical hardware impairments and formulates a design that maximizes the average achievable sum rate in multiuser \ac{ris}-aided \ac{cr} systems.
Finally, in \cite{Xu2025CR}, mobile-sensor \ac{sinr} values are maximized via joint optimization of the \ac{ss} time, \ac{sb} beamforming, and \ac{ris} beamforming, subject to \ac{ss} and secondary-transmission constraints.

\subsubsection{RIS- and SIM-aided Localization and Sensing}
Recent studies have investigated \ac{ris}-assisted localization \cite{fascista2022ris,fadakar2024multi,fadakar2025mutual,fadakar2025near}. 
In \cite{fascista2022ris}, a codebook-based joint localization and synchronization scheme is developed by jointly optimizing the \ac{bs} active precoder and the \ac{ris} passive phase profiles. 
The approach in \cite{fadakar2024multi} proposes a two-stage, low-complexity solution for joint localization and synchronization using multiple \acp{ris}, where a deep-learning-based \ac{aod} estimator is employed in the first stage. 
The works in \cite{fadakar2025mutual,fadakar2025near} consider joint estimation of the \ac{ue} position and mutual-coupling coefficients in multipath environments, addressing the far-field and near-field cases, respectively.

Wireless sensing using \acp{sim} has recently attracted interest \cite{An2024SIM_DOA,Wang2024SIM_ISAC,Li2025ISAC_SIM}. In \cite{An2024SIM_DOA}, an \ac{sim} is exploited for uplink \ac{aoa} estimation by configuring the stacked metasurfaces to implement a 2D \ac{dft} on the incident field. 
As waves propagate through the \ac{sim}, the receiver array directly observes the angular spectrum, and \ac{aoa} is estimated by probing the resulting energy distribution across the array.
In \cite{Wang2024SIM_ISAC}, the \ac{bs} beamformer and the \ac{sim}'s end-to-end transmission matrix are jointly optimized to estimate the complete target-response matrix, while satisfying minimum \ac{sinr} at communication users and transmit-power limits. 
Furthermore, \cite{Li2025ISAC_SIM} proposes a transmit-beamforming design for \ac{isac} using a \ac{sim} by performing beamforming in the passive wave domain, which reduces hardware complexity and power consumption.

\subsubsection{Gaps and Motivation}
The major gaps in the existing literature can be summarized as follows:
\begin{itemize}
\item 
Although the benefits of \ac{ris} in \ac{cr} networks have been demonstrated \cite{Yuan2025TTR_CR,Dong2025joint_secure,Yuan2021IRS_CR,Zhong2022DRL_IRS_CR,Ge2024RIS_CSS_CR,Hui2024RIS_CR_HWI,Zhao2023RIS_CR,Wu2022secure_RIS_CR,Xu2025CR}, the existing works are mostly communication-centric and \ac{isac} applications within \ac{cr} particularly localization remain underexplored.
\item 
Current \ac{sim} research is predominantly single-functional, focusing on either communication or sensing, and thus the potential of \acp{sim} for multi-functional tasks such as \ac{isac}, localization, and \ac{cr} has received little attention.
\item 
Algorithmic and beamforming solutions developed for single-layer \ac{ris} do not generalize to multi-layer \ac{sim} architectures. 
In particular, to the best of the authors' knowledge, the advantages of \acp{sim} in \ac{cr} has not yet been studied, motivating the development of novel \ac{sim}-specific algorithms and designs.
\end{itemize}

Motivated by these gaps, this paper studies a multicarrier wideband \ac{cr}-\ac{isac} system under the underlay scenario and optimizes the \ac{sim} coefficients to maximize \ac{su} localization performance while ensuring the average \ac{it} at \ac{pu}s remains below a predefined threshold at each \ac{ofdm} subcarrier.

\subsection{Contributions}
The main contributions of this paper are summarized as follows:
\begin{itemize}
\item 
We study a \ac{cr}-\ac{isac} system in a wideband \ac{ofdm} scenario assisted by a multi-functional \ac{sim}. 
Specifically, we formulate an optimization problem to design the \ac{sim} coefficients that maximize \ac{su} downlink localization performance while satisfying \ac{qos} requirement imposed on the \acp{pu}.
\item 
We derive the \ac{fim} for the downlink observations and obtain the corresponding \ac{bcrb} assuming a prior on the \ac{su} position. The \ac{bcrb} is then used as the localization metric to design the \ac{sim} coefficients.
\item 
We formulate the \ac{sim}-coefficient design as an optimization problem that minimizes the \ac{bcrb} of the \ac{su} position, subject to a constraint that the average \ac{se} degradation at the \acp{pu} does not exceed a predefined threshold. 
The resulting nonconvex problem is converted to a tractable max-min problem via convex optimization techniques.
\item 
We develop an efficient \ac{ao}-based algorithm to solve the max-min reformulation and obtain the optimal end-to-end transmission response of the \ac{sim} for all \ac{ofdm} subcarriers. 
Specifically, the inner minimization is solved via a bisection-based method, while a closed-form solution is derived for the outer maximization.
\item 
Exploiting the analogy between \ac{sim} architectures and \acp{dnn}, we first derive the analytical gradients required for backpropagation and then propose a learning-based scheme that uses the Adam optimizer to tune the \ac{sim} coefficients. 
A beampattern-based loss function is defined to drive the learned end-to-end response toward the previously obtained optimal solution. 
Extensive simulations demonstrate convergence of the proposed algorithms and the training procedure.
Moreover, the effectiveness of the proposed methods are validated for \ac{su} localization, \ac{it} management, or equivalently preservation of average \ac{se} at the \acp{pu}.
\end{itemize}

\paragraph*{Notation}
Throughout this paper, boldface uppercase letters (e.g., $\mathbf{X}$) denote matrices and boldface lowercase letters (e.g., $\mathbf{x}$) denote vectors; the superscripts $(\cdot)^{\mathsf{T}}$, $(\cdot)^{\mathsf{H}}$, and $(\cdot)^{-1}$ indicate transpose, Hermitian (conjugate transpose), and matrix inverse, respectively. 
We use $\Re\{\cdot\}$ and $\Im\{\cdot\}$ to denote the real and imaginary parts of a complex-valued quantity, respectively.
Horizontal concatenation is written as $[\mathbf{x}_1,\ldots,\mathbf{x}_n]$, $\mathrm{diag}(\mathbf{x})$ denotes the diagonal matrix whose main diagonal entries are those of $\mathbf{x}$, and $\mathbf{I}_n$ is the $n\times n$ identity matrix. We use $\lVert\mathbf{x}\rVert$ for the Euclidean norm and $\lVert\mathbf{A}\rVert_F$ for the Frobenius norm. 
$\mathbf{A}\odot\mathbf{B}$ an $\mathbf{A}\otimes\mathbf{B}$ denote the Hadamard and Kronecker products, respectively. 
The function $\operatorname{arctan2}(x,y)$ refers to the four-quadrant inverse tangent that returns the principal angle of the point $(x,y)$ in $(-\pi,\pi]$. For a real vector $\mathbf{x}=[x_1,\dots,x_n]^{\mathsf{T}}$, $e^{\mathbf{x}}$ denotes the elementwise exponential $[e^{x_1},\dots,e^{x_n}]^{\mathsf{T}}$. 
The inner product between two matrices is written as $\mathbf{A}\bullet\mathbf{B}=\mathrm{tr}(\mathbf{A}^{\mathsf{H}}\mathbf{B})$.
Finally, $\lambda_{\mathrm{max}}(\mathbf{A})$ denotes the largest eigenvalue of $\mathbf{A}$.

\vspace{-0.2cm}
\section{System Model} 
The \ac{cr}-\ac{isac} architecture under study is illustrated in Fig.~\ref{fig:system-model}. 
A single‐antenna \ac{pb} serves $N_{\text{pu}}$ single‐antenna \acp{pu}\footnote{The choice of a single-antenna \ac{pb} and \acp{pu} is made solely for notational simplicity and exposition. 
No assumptions are imposed on the internal structure of the \ac{pws}. 
Extension of the \ac{pws} components to multi-antenna architectures is straightforward and does not alter the proposed methodologies.} in the downlink, while a single‐antenna \ac{sb} integrated with an 
\ac{sim}, thus creating a reconfigurable antenna, simultaneously provides sensing via dedicated downlink pilot transmissions for \acp{su} while keeping the average \ac{se} of the \acp{pu} above a threshold. 
We assume that the relevant \ac{csi} namely, the \ac{pb}-\acp{pu} and \ac{sb}-\acp{pu} channels have been estimated and exchanged over a coordinated \ac{cr} control unit \cite{Xu2025CR}. 


\begin{figure}[!h]
\centering
\includegraphics[width=0.8\columnwidth]{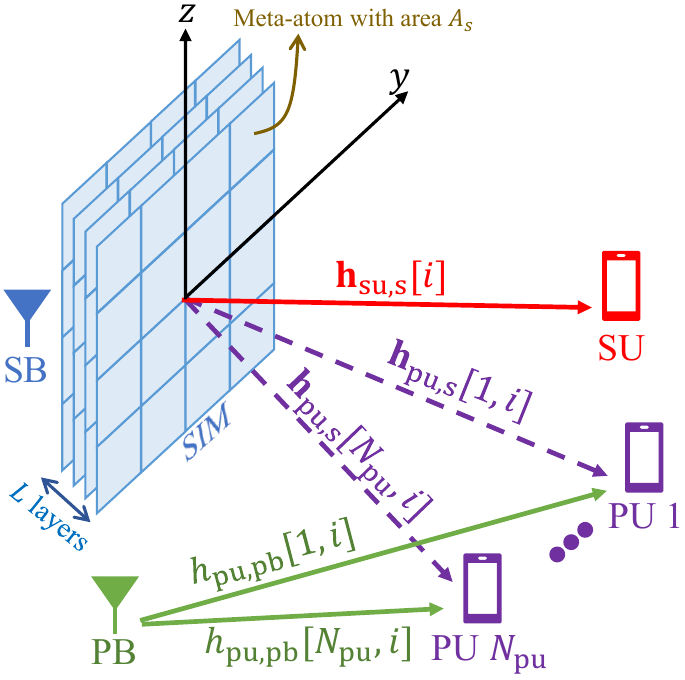}
\caption{
The considered \ac{sim}-aided \ac{cr}-\ac{isac} system. 
We aim to optimize the \ac{sim} coefficients to maximize the localization performance at the \ac{su} while keeping the averaged \ac{se} at the \acp{pu} above a threshold.
}
\label{fig:system-model}
\end{figure}

The \ac{sb} emits $I$ \ac{ofdm} downlink symbols through the \ac{sim} for localizing a single‐antenna \ac{su}. 
Moreover, the \ac{sim} is connected to a smart controller, such
as a customized \ac{fpga}, which is capable of independently
adjusting the phase shift of the \ac{em} waves transmitted through
each meta-atom \cite{An2025sim_downlink}.
Our goal is to optimize the \ac{sim} phase coefficients so as to maximize localization performance, while keeping the resulting interference leakage to the \acp{pu} under a specific limit. 
Because the downlink localization procedure can be carried out independently at each \ac{su}, in this paper we focus on the single \ac{su} case without loss of generality. 

\vspace{-0.5cm}
\subsection{Signal Model}
\subsubsection{SIM Signal Models}
The considered \ac{sim} consists of an
array of $L$ metasurfaces, each having a \ac{upa} shaped containing $N=N^hN^v$ meta-atoms with $N^h$ rows and $N^v$ columns. 
The phase shift matrix for the $\ell$-th \ac{sim} layer is denoted by 
\begin{equation}
\boldsymbol{\Phi}^{(\ell)}
=
\mathrm{diag}(
[e^{j\phi_1^{(\ell)}}, \dots, e^{j\phi_N^{(\ell)}}]^\mathsf{T}
)\in \mathbb{C}^{N\times N},
\end{equation}
where $\phi_n^{(\ell)}$ denotes the phase shift of the $n$-th meta-atom on the $\ell$-th metasurface layer.
For $\ell=2,\dots ,L$, let $\mathbf{W}^{(\ell)}[i]\in\mathbb{C}^{N\times N}$ represent the transmission matrix from the $(l-1)$-th to the $l$-th metasurface layer in the $i$-th subcarrier. 
According to the Rayleigh-Sommerfeld diffraction theory, the $(n,n')$-th entry of the wideband transmission matrix of the \ac{sim} is expressed as\cite{An2025sim_downlink, An2024SIM_DOA, An2023SIM_Holog_MIMO}:
\begin{equation}\label{eq:Sommerfeld_formula}
[\mathbf{W}^{(\ell)}[i]]_{n,n'}
=
\frac{
A_sd_s
}{(d_{n,n'}^{(\ell)})^2}
\left(
\frac{1}{2\pi d_{n,n'}^{(\ell)}}
-
j\frac{f_i}{c}
\right)
e^{j2\pi d_{n,n'}^{(\ell)}f_i/c},
\end{equation}
where $A_s$ is the area of each meta-atom, $d_s$ is the inter-layer distance in the \ac{sim}, $f_i=f_c-(i-1)\Delta f$ is the frequency of the $i$-th subcarrier (with $f_c$ the carrier frequency and $\Delta f$ the subcarrier spacing), and $c$ is the speed of light.
The inter-layer distance between the $n$-th meta-atom of the $l$-th layer and the $n'$-th meta-atm of the $(l-1)$-th layer is represented as $d_{n,n'}^{(\ell)}$.
Finally, the cumulative effect of signal propagation through the layers
of the \ac{sim} in the $i$-th subcarrier is characterized by $\mathbf{G}[i]\in \mathbb{C}^{N\times N}$ obtained as \cite{li2025stacked_ofdm, An2025sim_downlink}:
\begin{equation}\label{eq:sim_end_to_end}
\mathbf{G}[i]
=
\boldsymbol{\Phi}^{(L)}
\mathbf{W}^{(L)}[i]
\dots 
\boldsymbol{\Phi}^{(2)}
\mathbf{W}^{(2)}[i]
\boldsymbol{\Phi}^{(1)}.
\end{equation}

\subsubsection{Received Downlink Signal Models}
The received signal at the \ac{su} and the $r$-th \ac{pu} over the $i$-th subcarrier are given by:
\begin{subequations}\label{def:received_sigs}
\begin{align}\label{eq:SUE-received-sig}
y_{\text{su}}[i]
& =
\sqrt{P_\text{sb}}
\mathbf{h}_{\text{su,s}}^{\mathsf{T}}[i]
\mathbf{G}[i]
\mathbf{w}[i] 
+
v_{\text{su}}[i],
\\
y_{\text{pu}}[r,i]
& =
\sqrt{P_\text{pb}}
h_{\text{pu,pb}}[r,i]
s[r,i] \notag \\
& +
\sqrt{P_\text{sb}}
\mathbf{h}_{\text{pu,s}}[r,i]^\mathsf{T}
\mathbf{G}[i]
\mathbf{w}[i]
+
v_{\text{pu}}[r,i],
\end{align}
\end{subequations}
where $\mathbf{w}[i] \in \mathbb{C}^{N}$ is the transmission vector from the \ac{sb} to the first layer of the \ac{sim}.
Moreover, $P_\text{sb}$ and $P_\text{pb}$ represent the total transmitted power from the \ac{sb} and \ac{pb}, respectively. 
The communication symbol transmitted from the \ac{pb} to the $r$-th \ac{pu} in the $i$-th subcarrier is denoted by $s[r,i]$, which satisfies the unit-power condition $\mathbb{E}[\lvert s[r,i]\rvert^2] = 1$.
$v_{\text{su}}[i] \sim \mathcal{CN}(0,\sigma^2)$ denotes the joint contribution of the \ac{awgn} and the interference from the \ac{pb} at the \ac{su}, and 
$v_{\text{pu}}[r,i] \sim \mathcal{CN}(0,\sigma^2_v)$ denote the \ac{awgn} at the $r$-th \ac{pu}.
Additionally, 
$\mathbf{h}_{\text{su,s}}[i]\in\mathbb{C}^{N}$, 
$h_{\text{pu,pb}}[r,i]\in\mathbb{C}$, and
$\mathbf{h}_{\text{pu,s}}[r,i]\in\mathbb{C}^{N}$
denote the \ac{su}-\ac{sim}, \ac{pu}-\ac{pb}, and \ac{pu}-\ac{sim} 
channels, respectively which will be explained in detail in the following subsection. 
\vspace{-0.2cm}
\subsection{Channel Model}
A widely used approach for representing the wideband \ac{ofdm} channel model is the frequency-domain form, which simplifies beamspace processing via the IFFT/FFT
\cite{li2025stacked_ofdm, Heath2016Overview}.
Assuming a block-fading geometric channel model \cite{You2021Wireless, Chen2024fluid} between any two components of the considered \ac{cr}-\ac{isac} system, the channels are modeled as follows:
\begin{subequations}\label{def:channels}
\begin{align}
\mathbf{h}_{\text{su,s}}[i]
& =
\overline{\mathbf{h}}_{\text{su,s}}[i]
+
\sum_{q=1}^{Q_{\text{su,s}}}
\tilde{\mathbf{h}}_{\text{su,s}}^{(q)}[i]
,
\label{ch:su_s}
\\
\mathbf{h}_{\text{pu,s}}[r,i]
& =
\overline{\mathbf{h}}_{\text{pu,s}}[r,i]
+
\sum_{q=1}^{Q_{\text{pu,s}}}
\tilde{\mathbf{h}}_{\text{pu,s}}^{(q)}[r,i],
\label{ch:pu_s}
\\
h_{\text{pu,pb}}[r,i]
& =
\overline{h}_{\text{pu,pb}}[r,i]
+
\sum_{q=1}^{Q_{\text{pu,pb}}}
\tilde{h}_{\text{pu,pb}}^{(q)}[r,i],
\label{ch:pu_pb}
\end{align}
\end{subequations}
where $\mathbf{h}_{\text{su,s}}[i]\in\mathbb{C}^{N}$ denotes the channel between the \ac{sim} and \ac{su}, 
$\mathbf{h}_{\text{pu,s}}[r,i]\in\mathbb{C}^{N}$ is the channel between the \ac{sim} and the $r$-th \ac{pu}, and 
$h_{\text{pu,pb}}[r,i]$ is the channel between the \ac{pb} and the $r$-th \ac{pu}. 
The index $i$ corresponds to the $i$-th subcarrier. 
Moreover, $\overline{\mathbf{h}}_{\text{su,s}}[i]$, $\overline{\mathbf{h}}_{\text{pu,s}}[r,i]$, and $\overline{h}_{\text{pu,pb}}[r,i]$ denote the \ac{los} terms of these channels, respectively.
Likewise, $\widetilde{\mathbf{h}}_{\text{su,s}}^{(q)}[i]$, $\widetilde{\mathbf{h}}_{\text{pu,s}}^{(q)}[r,i]$, and $\widetilde{h}_{\text{pu,pb}}^{(q)}[r,i]$ represent the \ac{nlos} components through the $q$-th \ac{mpc}. 
Finally, $Q_{\text{su,s}}$, $Q_{\text{pu,s}}$, and $Q_{\text{pu,pb}}$ are the number of such \ac{nlos} paths in the corresponding channels. 

The \ac{los} terms of the channels are defined as follows:
\begin{subequations}\label{def:los_channels}
\begin{align}
\label{def:los_h_su_s}
\overline{\mathbf{h}}_{\text{su,s}}[i]
& =
\alpha_{\text{su,s}}
e^{-j \zeta_i \tau_{\text{su,s}}}
\mathbf{a}_{i}(\boldsymbol{\theta}_{\text{su,s}}),
\\ 
\overline{\mathbf{h}}_{\text{pu,s}}[r,i]
& =
\alpha_{\text{pu,s}}[r]
e^{-j \zeta_i \tau_{\text{pu,s}}[r]}
\mathbf{a}_{i}(\boldsymbol{\theta}_{\text{pu,s}}[r]),
\\
\overline{h}_{\text{pu,pb}}[r,i]
& =
\alpha_{\text{pu,pb}}[r,i]
e^{-j \zeta_i \tau_{\text{pu,pb}}[r]},
\end{align}
\end{subequations}
where 
$\alpha_{\text{su,s}}$, $\alpha_{\text{pu,s}}[r]$, and $\alpha_{\text{pu,pb}}[r,i]$, denote the corresponding overall complex gain of the channels.
In addition, $\zeta_i=2\pi\,(i-1)\Delta f$, and the delay $\tau_{u,v}$ between two components $u,v$ is defined as 
$\tau_{u,v}
=
\lVert \mathbf{p}_{u}-\mathbf{p}_{v}\rVert / c$.
In \eqref{def:los_channels}, $\boldsymbol{\theta}_\text{su,s}=[\theta_\text{su,s}^{\text{el}}, \theta_\text{su,s}^{\text{az}}]^\mathsf{T}$ and $\boldsymbol{\theta}_\text{pu,s}[r]=[\theta_\text{pu,s}^{\text{el}}[r], \theta_\text{pu,s}^{\text{az}}[r]]^\mathsf{T}$ denote the 2D-\ac{aod} from the last layer of the \ac{sim} towards the the \ac{su} and the $r$-th \ac{pu}, respectively. 
The angles
$\theta_\text{su,s}^{\text{el}}, \theta_\text{su,s}^{\text{az}}$ are related to the positions of these components as follows:
\begin{align}\label{eq:theta_phi_vals}
& 
\theta_\text{su,s}^{\text{el}}
=
\text{arccos} \left(
\frac{[\mathbf{p}_{\text{su;s}}]_{3}}{\lVert\mathbf{p}_{\text{su;s}}\rVert}
\right),\ 
\theta_\text{su,s}^{\text{az}}
=
\text{arctan2} (
[\mathbf{p}_{\text{su;s}}]_{2},[\mathbf{p}_{\text{su;s}}]_{1}
),
\end{align}
where 
$\mathbf{p}_\text{su;s}=\mathbf{R}_\text{s}^\mathsf{T}(\mathbf{p}_\text{su}-\mathbf{p}_\text{s})$ 
denotes the local coordinates of 
$\mathbf{p}_\text{su}$ 
with respect to the \ac{sim}. 
Similar equations can be obtained for $\theta_\text{pu,s}^{\text{el}}[r], \theta_\text{pu,s}^{\text{az}}[r]$.
The \ac{arv} of the last layer of the \ac{sim} at the $i$-th subcarrier, denoted by $\mathbf{a}_i(\boldsymbol{\theta})$, is defined as:
\begin{equation}
\mathbf{a}_i(\boldsymbol{\theta})
=
e^{-j2\pi\,\omega_i^{h}\,\mathbf{k}(N^h)}
\otimes
e^{-j2\pi\,\omega_i^{v}\,\mathbf{k}(N^v)},
\end{equation}
where $\mathbf{k}(N) = [0, \dots, N-1]^\mathsf{T}$, and $\omega_i^{h}$ and $\omega_i^{v}$ represent the spatial horizontal and vertical frequencies for the $i$-th subcarrier, respectively.  
As shown in Fig.~\ref{fig:system-model}, if the \ac{sim} \ac{upa}-shaped layers are positioned parallel to the $y$-$z$ plane of the local Cartesian coordinate system of the \ac{sb}, then  
\begin{equation}
\omega_i^{h}
=
d
\sin(\theta^{\text{az}})
\sin(\theta^{\text{el}})/\lambda_i
,\ 
\omega^{v}
=
d
\cos(\theta^{\text{el}})/\lambda_i,
\end{equation}
where $d$, denotes the distance between two adjacent meta-atoms in each \ac{sim} layer, and $\lambda_i=c/f_i$ is the wavelength of the $i$-th subcarrier frequency. 
Moreover, $\theta^{\mathrm{az}},\theta^{\mathrm{el}}$ represent the azimuth and elevation angles in the \ac{sb}’s local Cartesian coordinates. 
Azimuth $\theta^{\mathrm{az}}\in[-\pi,\pi)$ is measured from $+y$ toward $+x$, and elevation $\theta^{\mathrm{el}}\in[0,\pi]$ is the angle with respect to $+z$ direction.
All parameters associated with the \ac{nlos} paths in \eqref{def:channels} are defined in a similar manner using single-point scatterers.
To focus on the main tasks of the paper, we assume that the channels at the \acp{pu} i.e., $h_{\text{pu,pb}}[r,i]$ and $\mathbf{h}_{\text{pu,s}}[r,i]$ are estimated and known at the \ac{sb} for the design of the \ac{sim} phase coefficients.

\vspace{-0.2cm}
\section{Lower Bound Analysis and Problem Formulation}
In this section, we perform \ac{fim} analysis to derive the \ac{crb} values for the channel and state parameters, which play a crucial role in designing the \ac{sim} phase coefficients. 
Note that the \ac{nlos} components via the \acp{mpc} in the \ac{su}-\ac{sim} channel of \eqref{ch:su_s} do not contribute to the theoretical localizability of the \ac{su} \cite{Chen2024Multi, fadakar2025hybrid, fadakar2025mutual}, and they may degrade performance by producing unresolvable paths.
Accordingly, the lower bounds presented in this paper are derived only considering the \ac{los} path between the \ac{sim} and the \ac{su}, rather than the full model in \eqref{eq:SUE-received-sig}, \eqref{ch:su_s}.
\subsubsection{CRB in Channel Domain}
In this subsection, we compute the standard \ac{fim} of the unknown channel parameter vector 
$\bm{\eta}=[
\theta_{\text{su,s}}^{\text{el}}, 
\theta_{\text{su,s}}^{\text{az}},
\tau_{\text{su,s}},
\rho_{\text{su,s}},\varphi_{\text{su,s}}
]^\mathsf{T}
$ where $\rho_{\text{su,s}},\varphi_{\text{su,s}}$ are the magnitude and phase components of the complex path gain $\alpha_{\text{su,s}}$. 
Since the observations in \eqref{eq:SUE-received-sig} are complex Gaussian, we utilize Slepian-Bangs formula to obtain the \ac{fim} $\mathbf{J}_{\bm{\eta}}\in\mathbb{R}^{5\times 5}$ \cite{fascista2022ris,fadakar2025hybrid}. 
Specifically, the $(u,w)$-th element can be computed as:
\begin{equation}\label{eq:FIM_ch_def}
[\mathbf{J}_{\bm{\eta}}]_{u,w}
=
\frac{2}{\sigma^2}
\sum_{i=1}^{I}
\Re
\bigg\{
\left(
\frac{\partial x_{\text{su}}[i]}{\partial [\bm{\eta}]_{u}}
\right)^{*}
\left(
\frac{\partial x_{\text{su}}[i]}{\partial [\bm{\eta}]_{w}}
\right)
\bigg\}
,
\end{equation}
where $x_{\text{su}}[i]$ denotes the noise- and multipath-free component of \eqref{eq:SUE-received-sig}, which can be represented as:
\begin{equation}
x_{\text{su}}[i]
=
\mathbf{c}[i]^\mathsf{T}
\mathbf{f}[i],
\end{equation}
where $\mathbf{f}[i]\in \mathbb{C}^N$ is the tunable end-to-end response of the \ac{sim} in the $i$-th subcarrier:
\begin{equation}\label{def:f}
\mathbf{f}[i]
=
\mathbf{G}[i]
\mathbf{w}[i],
\end{equation}
and 
$\mathbf{c}[i]\in \mathbb{C}^N$ is an uncontrollable vector proportional to the \ac{los} channel.
are defined as follows:
\begin{align}
\mathbf{c}[i]
& =
\sqrt{P_\text{sb}}\,
\overline{\mathbf{h}}_{\text{su,s}}[i]
=
\sqrt{P_\text{sb}}
\alpha_{\text{su,s}}
e^{-j \zeta_i \tau_{\text{su,s}}}
\mathbf{a}_{i}(\boldsymbol{\theta}_{\text{su,s}}).
\end{align}
Thus, the \ac{fim} entries in \eqref{eq:FIM_ch_def} can be represented as:
\begin{align}\label{eq:fim_ch_detailed}
[\mathbf{J}_{\bm{\eta}}]_{u,w}
=
\frac{2}{\sigma^2}
\sum_{i=1}^{I}
\Re\big\{
\mathbf{f}[i]^{\mathsf{H}}
\mathbf{c}_{u}[i]^*
\mathbf{c}_{w}[i]^\mathsf{T}
\mathbf{f}[i]
\big\},
\end{align}
where 
$
\mathbf{c}_{u}
=
\partial \mathbf{c}
/
\partial [\bm{\eta}]_{u}.
$
\subsubsection{CRB in State Domain}
To obtain the \ac{fim} $\mathbf{J}_{\bm{\gamma}}\in\mathbb{R}^{5\times 5}$ for the state domain vector 
$\bm{\gamma}=[\mathbf{p}_\text{su}^\mathsf{T},
\rho_{\text{su,s}},\varphi_{\text{su,s}}
]^\mathsf{T}$, 
we use the Jacobian transformation matrix $\mathbf{T}=\frac{\partial \bm{\eta}^\mathsf{T}}{\partial \bm{\gamma}}\in \mathbb{R}^{5\times 5}$ as 
$\mathbf{J}_{\bm{\gamma}}=\mathbf{T}\mathbf{J}_{\bm{\eta}}\mathbf{T}^\mathsf{T}$. 
Thus, using \eqref{eq:fim_ch_detailed}, 
the $(u,w)$-th element of $\mathbf{J}_{\bm{\gamma}}$ can be represented as:
\begin{align}\label{eq:fim_state_detailed}
[\mathbf{J}_{\bm{\gamma}}]_{u,w}
=
\frac{2}{\sigma^2}
\sum_{i=1}^{I}
\Re\big\{
\mathbf{f}[i]^{\mathsf{H}}
\tilde{\mathbf{c}}_{u}[i]^*
\tilde{\mathbf{c}}_{w}[i]^\mathsf{T}
\mathbf{f}[i]
\big\}
\end{align}
where
\begin{align}
\widetilde{\mathbf{c}}_{u}[i]
=
\sum_{j=1}^{5} 
[\mathbf{T}]_{u,j}
\mathbf{c}_{j}[i].
\end{align}

\subsubsection{Bayesian CRB in State Domain}
Assuming a prior \ac{pdf} $p(\boldsymbol{\gamma})\in\mathbf{R}^{5}$ 
for the state parameter vector $\boldsymbol{\gamma}$, 
the \ac{bfim} is defined as follows:
\begin{equation}\label{def:J_B}
\mathbf{J}_{B}
=
\mathbf{J}_{D}
+
\mathbf{J}_{P},
\end{equation}
where
\begin{equation}
\mathbf{J}_{D}
=
\mathbb{E}_{\boldsymbol \gamma}\bigl[\mathbf J_{\boldsymbol \gamma}(\boldsymbol \gamma)\bigr],
\end{equation}
whose $(u,w)$-th entry can be obtained by applying expectation on \eqref{eq:fim_state_detailed}:
\begin{align}\label{eq:FIM_entry}
[\mathbf{J}_{D}]_{u,w}
=
&
\sum_{i=1}^{I}
\Re\bigg\{
\mathbf{f}[i]^{\mathsf{H}}
\mathbb{E}_{\boldsymbol \gamma}
\bigg[
\frac{2}{\sigma^2}
\widetilde{\mathbf{c}}_{u}^*[i]
\widetilde{\mathbf{c}}_{w}^\mathsf{T}[i]
\bigg]
\mathbf{f}[i]
\bigg\},
\end{align}
and 
\begin{equation}
\mathbf{J}_{P}
=
\mathbb{E}_{\bm{\gamma}}\!
\biggl[
-\,
\nabla_{\bm\gamma}^2
\,\ln p(\bm\gamma)
\biggr],
\end{equation}
where $\nabla_{\bm\gamma}^2\triangleq \frac{\partial^2}{\partial\bm\gamma\,\partial\bm\gamma^{T}}$ denotes the Hessian matrix operator. 
Finally, the \ac{su} position \ac{bcrb} is defined as: 
\begin{equation}
\mathrm{BCRB}
=
\mathrm{tr}\bigl([\mathbf J_{B}^{-1}]_{1:3,\,1:3}\bigr)
.
\end{equation}

\subsection{Average SINR at the PUs}
The average \ac{sinr} at the \ac{pu}s on the $i$-th subcarrier is defined as:
\begin{equation}\label{eq:avg_SINR}
\mathrm{SINR}[i]
=
\frac{
S^2[i]
}{
I^2[i]
+
\sigma^2_v
},\ 
\overline{\mathrm{SINR}}[i]
=
\frac{
S^2[i]
}{
\sigma^2_v
},
\end{equation}
where $\mathrm{SINR}[i]$ denotes the \ac{sinr} in presence of interference, and $\overline{\mathrm{SINR}}[i]$ represents the same value in absence of the interference from the \ac{sim}. 
Moreover, $S^2[i]$ and  $I^2[i]$ denote the average signal and interference power, respectively, at the \ac{pu}s on the $i$-th subcarrier, which are given by:
\begin{align}
S^2[i]
&
=
\frac{P_\text{pb}}{N_{\text{pu}}}
\sum_{r=1}^{N_\text{pu}}
\lvert 
h_{\text{pu,pb}}[r,i]
\rvert^2
\\
I^2[i]
& =
\frac{P_\text{sb}}{N_{\text{pu}}}
\sum_{r=1}^{N_\text{pu}}
\lVert
\mathbf{h}_{\text{pu,s}}[i]^\mathsf{T}
\mathbf{f}[i]
\rVert^2
=
\mathbf{f}[i]^\mathsf{H}
\mathbf{R}_\text{pu}
\mathbf{f}[i],
\end{align}
where 
\begin{equation}
\mathbf{R}_\text{pu}[i]
=
\frac{P_\text{sb}}{N_\text{pu}}
\sum_{r=1}^{N_\text{pu}}
\mathbf{h}^*_{\text{pu,s}}[r,i]
\mathbf{h}_{\text{pu,s}}[r,i]^\mathsf{T}.
\end{equation}

\vspace{-0.5cm}
\subsection{Problem Formulation}\label{sec:problem-formulation}
The design of the phase shift coefficients of the \ac{sim} layers must simultaneously:
(i) minimize the \ac{su} \ac{bcrb}, and (ii) keep the averaged \ac{se} of the \ac{pu}s above a certain threshold. 
To this end, the optimization problem can be formulated as follows:
\begin{subequations}\label{eq:opt_formulation}
\begin{align}
\min_{\{\boldsymbol{\Phi}^{(\ell)}\}_{\ell=1}^{L}}\;&
\text{BCRB}(
\{\boldsymbol{\Phi}^{(\ell)}\}_{\ell=1}^{L}
)
\\
\text{s.t.}\;
&
\mathrm{log}_2(1+\mathrm{SINR}[i])
\ge 
\kappa\,
\mathrm{log}_2(1+\overline{\mathrm{SINR}}[i])
\label{const:pu_interference}
,\\
& 
i=1,\dots ,I\, ,
\notag
\\
&\lvert[\boldsymbol{\Phi}^{(\ell)}]_{n,n}\rvert = 1,
\label{const:sim_unit}
\\
&
\forall\,n=1,\dots,N\, ,
\forall\,l=1,\dots,L\, ,
\notag
\end{align}
\end{subequations}
where \eqref{const:pu_interference} ensures that the average data rate at the \ac{pu}s is at least $\kappa$ times the average data rate when \ac{sb} is not transmitting,
and \eqref{const:sim_unit} imposes the unit‐modulus constraint on the \ac{sim} layer coefficients.  
\vspace{-0.5cm}
\section{Proposed Beamforming Method}\label{sec:beamforming}
Problem \eqref{eq:opt_formulation} is non-convex and high-dimensional and challenging due to coupled phase coefficients of the \ac{sim} layers. 
To make the problem tractable, we first assume a full freedom over the end-to-end \ac{sim} response vectors $\mathbf{f}[i]$ and drop the condition \eqref{def:f}. 
After optimizing the vectors $\mathbf{f}[i]$, we will optimize the \ac{sim} coefficients in Sec.~\ref{sec:sim_opt}. 
Thus, the relaxed problem can be reformulated as follows:
\begin{subequations}\label{eq:opt_formulation_f}
\begin{align}
\min_{\{\mathbf{f}[i]\}_{i=1}^I}\;&
\mathrm{BCRB}(\{\mathbf{f}[i]\}_{i=1}^I)
\\
&
\text{s.t.}\;
\epsilon[i]
\ge
\mathbf{f}[i]^\mathsf{H}
\mathbf{R}_{\text{pu}}[i]
\mathbf{f}[i],\ 
i=1,\dots, I
\label{const:pu_rate}
,\\
& 
\lVert
\mathbf{f}[i]
\rVert^2
\le \delta[i],
\label{const:F_pow}
\end{align}
\end{subequations}
where \eqref{const:pu_rate} is equivalent to \eqref{const:pu_interference} where $\epsilon[i]$ is defined as:
\begin{equation}\label{def:epsilon}
\epsilon[i]
=
\frac{S^2[i]-R[i]\sigma^2}{R[i]}
\end{equation}
where $R[i]=(1+\overline{\mathrm{SINR}}[i])^{\kappa}-1$, and, since $\mathbf{f}[i]$ is upper-bounded according to \eqref{def:f}, constraint \eqref{const:F_pow} enforces the per-subcarrier power limit $\delta[i]$ on $\mathbf{f}[i]$.
\begin{proposition}\label{prop:problem_equivalence}
The problem \eqref{eq:opt_formulation_f} is equivalent to:
\begin{subequations}\label{opt:equivalent_v1}
\begin{align}
\max_{\{\mathbf{d}_j\}_{j=1}^3}
& \min_{\{\mathbf{f}[i]\}_{i=1}^I} 
\sum_{j=1}^{3}
\left(
-
\mathbf{d}_{j}^\mathsf{T}
\mathbf{J}_B
\mathbf{d}_{j}
+
2\mathbf{d}_{j}^\mathsf{T}\mathbf{e}_j
\right)
\label{obj:equivalent_v1}
\\
&
\text{s.t.}\;
\eqref{const:pu_rate}, \eqref{const:F_pow},
\notag
\end{align}
\end{subequations}
where $\mathbf{d}_{j}\in\mathbb{R}^{5}$ for $j=1,2,3$ are optimization variables for the outer maximization subproblem. 
This problem can be also represented as:
\begin{subequations}\label{opt:equivalent_v2}
\begin{align}
\max_{\{\mathbf{d}_j\}_{j=1}^3}
\min_{\{\mathbf{f}[i]\}_{i=1}^I} 
& 
\bigg[
\sum_{j=1}^{3}
\left(
-
\mathbf{d}_{j}^\mathsf{T}
\mathbf{J}_P
\mathbf{d}_{j}
+
2\mathbf{d}_{j}^\mathsf{T}\mathbf{e}_j
\right)
\notag \\
& -
\sum_{i=1}^{I}
\Re\big\{
\mathbf{f}[i]^{\mathsf{H}}
\overline{\mathbf{C}}[i]
\mathbf{f}[i]
\big\}
\bigg]
\\
&
\text{s.t.}\;
\eqref{const:pu_rate}, \eqref{const:F_pow},
\notag
\end{align}
\end{subequations}
where $\mathbf{e}_j$ is the $j$-th column of the identity matrix $\mathbf{I}_5$ and:
\begin{equation}
\overline{\mathbf{C}}[i]
=
\sum_{j=1}^{3}
\mathbb{E}
\bigg[
\frac{2}{\sigma^2}
\overline{\mathbf{c}}_{\mathbf{d}_j}[i]^*
\overline{\mathbf{c}}_{\mathbf{d}_j}[i]^\mathsf{T}
\bigg]
\end{equation}
where 
\begin{equation}
\overline{\mathbf{c}}_{\mathbf{d}_j}[i]
=
\sum_{r=1}^{5}
[\mathbf{d}_j]_r
\tilde{\mathbf{c}}_{r}[i]
\end{equation}
\end{proposition}
\begin{proof}
See Appendix~\ref{app:FIM-derivatives}.
\end{proof}
\vspace{-0.1cm}
Hence, in order to solve the original problem \eqref{eq:opt_formulation}, we first solve the equivalent problem \eqref{opt:equivalent_v1} or \eqref{opt:equivalent_v2}, and then we optimize the \ac{sim}'s coefficients in the next section. 
To tackle the max-min problem \eqref{opt:equivalent_v1}, we resort to \ac{ao} algorithm to iteratively solve the inner min and outer max problems until convergence whose details are explained next.

\vspace{-0.3cm}
\subsection{Inner Optimization}\label{sec:inner_opt}
For any subcarrier index $i$, the optimal vector $\hat{\mathbf{f}}[i]$ can be obtained using the following proposition:

\begin{proposition}\label{prop:inner_opt}
Since for each subcarrier the inner optimization problem in \eqref{opt:equivalent_v2} is solved independently from other subcarriers, we drop the subcarrier index $i$ for notational simplicity. 
Let $\mathbf{A} = \tfrac{1}{2}(\overline{\mathbf{C}}+\overline{\mathbf C}^{\mathsf{H}})$.
Then, assuming fixed vectors $\{\mathbf{d}_j\}_{j=1}^{3}$, the optimal solution to the inner optimization can be obtained using Algorithm~\ref{alg:F_opt}.
\end{proposition}
\begin{proof}
See Appendix~\ref{app:inner_opt}.
\end{proof}

\begin{algorithm}[ht] 
\caption{Proposed Method for the Inner Optimization \eqref{opt:equivalent_v2}}\label{alg:F_opt}
\begin{algorithmic}[1] 
\State 
\multiline{%
\textbf{Inputs}: 
Matrix $\mathbf{A}=\tfrac{1}{2}(\overline{\mathbf{C}}+\overline{\mathbf C}^{\mathsf{H}})$, covariance matrix $\mathbf R$, scalars $\kappa,\delta>0$, and the tolerance $\xi_\mathrm{tol}>0$. 
}
\State 
\multiline{%
Obtain $\epsilon$ based on $\kappa$ using \eqref{def:epsilon}.
}
\State \textbf{If} $\lambda_{\max}(\mathbf A)\le 0$ \textbf{then return} $\hat{\mathbf{f}}=\mathbf{0}$ \textbf{End If}
\State 
\multiline{%
Compute principal unit eigenvector $\mathbf{u}$ of $\mathbf A$. 
}
\State \textbf{If}
$\mathbf{u}^{\mathsf{H}}\mathbf{R}\mathbf{u} \le \epsilon/\delta$
\textbf{then return}
$\hat{\mathbf{f}}=\sqrt{\delta}\,\mathbf u$. 
\State 
\textbf{Else} We need to find $\mu>0$ s.t. $g(\mu)\triangleq \mathbf{v}(\mu)^{\mathsf{H}} 
\mathbf{R} 
\mathbf{v}(\mu) =\epsilon/\delta$ 
where $\mathbf{v}(\mu)$ is the principal eigenvector of 
$
\mathbf{A} - \mu\mathbf{R}
$
\State 
\textbf{Perform bisection algorithm to find $\mu$:}
\State 
\multiline{%
$\mu_{\ell}\leftarrow 0$, $\mu_{u}\leftarrow 1$ 
}
\State
\textbf{While} $g(\mu_{u})>\epsilon/\delta$ \textbf{do} 
$\mu_{u}\leftarrow 2\mu_{u}$
\textbf{end while}
\Repeat 
\State 
$\mu\leftarrow (\mu_{\ell}+\mu_{u})/2$ 
\State 
\multiline{%
compute principal unit eigenvector $\mathbf{v}(\mu)$ of $\mathbf{A}-\mu\mathbf{R}$
}
\State 
$g(\mu) \leftarrow \mathbf{v}(\mu)^{\mathsf{H}}\mathbf{R}\mathbf{v}(\mu)$ 
\State
\textbf{If} $g(\mu)>\epsilon/\delta$
\textbf{then} $\mu_{\ell}\leftarrow\mu$ \textbf{else}  $\mu_{u}\leftarrow\mu$ 
\textbf{end if} 
\Until{$\lvert g(\mu)-\epsilon/\delta\rvert<\xi_\mathrm{tol}$} 
\State 
\textbf{Output}:
$\hat{\mathbf{f}}=\sqrt{\delta}\,\mathbf{v}(\mu)$. 
\end{algorithmic} 
\end{algorithm}

\subsection{Outer Optimization}\label{sec:outer_opt}
Assuming fixed vectors $\{\mathbf{f}[i]\}_{i=1}^{I}$, the outer maximization problem in \eqref{opt:equivalent_v1} can be represented as:
\begin{align}\label{eq:outer_problem_rev}
\max_{\{\mathbf d_j\}_{j=1}^3}\;
\sum_{j=1}^{3}\Big(
-
\mathbf d_j^\mathsf{T}\mathbf J_B\mathbf d_j
+
2\,\mathbf{e}_j^\mathsf{T}\mathbf{d}_j
\Big).
\end{align}

\begin{proposition}\label{prop:outer}
The problem \eqref{eq:outer_problem_rev} can be solved independently for each $j=1,2,3$, and each $\mathbf{d}_j\in\mathbb R^{5}$ is uniquely obtained as:
\begin{equation}\label{eq:d_star_5}
\hat{\mathbf{d}}_j
= 
\mathbf J_B^{-1}\mathbf e_j,\qquad j=1,2,3.\;
\end{equation}
\end{proposition}
\begin{proof}
See Appendix~\ref{app:outer_opt}.
\end{proof}

\subsection{Alternating optimization}
In the previous subsections, we solved the inner and outer optimization subproblems in the max-min problem \eqref{opt:equivalent_v1} independently. 
In this subsection we propose an \ac{ao} based algorithm to solve the entire problem.
To this end, we iteratively optimize the vectors $\{\mathbf{f}[i]\}_{i=1}^{I}$ and $\{\mathbf{d}_j\}_{j=1}^{3}$ using Algorithm~\ref{alg:F_opt} and \eqref{eq:d_star_5}, respectively until convergence. 
Details are provided in Algorithm~\ref{alg:alt_opt}.
Note that in Prop.~\ref{prop:inner_opt} we solved the problem using fixed vectors $\{\mathbf{d}_j\}_{j=1}^{3}$. 
Thus, it is essential to obtain a robust initialization for these vectors. 
First, we provide the following lemma:
\begin{lemma}\label{lemma:trace}
A lower bound for \ac{bcrb} can be obtained as:
\begin{equation}
\mathrm{BCRB}\ge \frac{9}{\mathrm{tr}([\mathbf{J}_B]_{1:3,1:3})}.
\end{equation}
\end{lemma}
\begin{proof}
See Appendix~\ref{app:lemma_trace}.
\end{proof}

From Lemma~\ref{lemma:trace}, it can been seen that $9/\mathrm{tr}([\mathbf{J}_B]_{1:3,1:3})$ is a lower bound for the \ac{bcrb}. 
Hence, maximizing $\mathrm{tr}([\mathbf{J}_B]_{1:3,1:3})$ would minimize this lower bound. Note that:
\begin{equation}\label{eq:trace_JB}
\mathrm{tr}([\mathbf{J}_B]_{1:3,1:3})
=
\sum_{j=1}^{3}
\mathbf{e}_j^\mathsf{T}
\mathbf{J}_B
\mathbf{e}_j.
\end{equation}
Comparing \eqref{eq:trace_JB} with \eqref{obj:equivalent_v1}, we observe that \eqref{eq:trace_JB} is the objective value of the inner optimization problem using the vectors $\mathbf{d}_j=\mathbf{e}_j$ for $j\in\{1,2,3\}$. 
Hence, the vector $\mathbf{e}_j$ is a suitable choice as initialization for the vector $\mathbf{d}_j$. 
\begin{algorithm}[ht]
\caption{Alternating optimization to solve \eqref{opt:equivalent_v1}}\label{alg:alt_opt}
\begin{algorithmic}[1]
\State \textbf{Inputs:} 
tolerances $\varepsilon_{\text{tol}},\tau_{\text{tol}}$.
\State 
\textbf{Initialization:}
$\mathbf{d}_j\gets \mathbf{e}_j$ for $j\in\{1,2,3\}$, 
$\mathbf{D}^{(0)}=[\mathbf{d}_1,\mathbf{d}_2,\mathbf{d}_3]$, $g^{(0)}\gets \infty$.
\While{$\big|g^{(k)}-g^{(k-1)}\big|/|g^{(k-1)}| > \varepsilon_{\text{tol}}$ \textbf{and} $\|\mathbf{D}^{(k)}-\mathbf{D}^{(k-1)}\|_F>\tau_{\text{tol}}$}
\State 
\multiline{%
Update $\{\mathbf{f}[i]\}_{i=1}^{I}$ via Algorithm~\ref{alg:F_opt}.

$g^{(k)}\gets $ objective value \eqref{obj:equivalent_v1}.
}
\State 
\multiline{%
Update $\{\mathbf{d}_j\}_{j=1}^{3}$ based on closed form solution \eqref{eq:d_star_5}.

$\mathbf{D}^{(k)}\gets [d_1,d_2,d_3]$.
}
\EndWhile
\State
\textbf{Outputs:} $\{\hat{\mathbf{f}}[i]\}_{i=1}^{I}$, $\{\hat{\mathbf{d}_j\}}_{j=1}^{3}$
\end{algorithmic}
\end{algorithm}

\section{Optimization of the SIM Coefficients}\label{sec:sim_opt}
In Sec.~\ref{sec:beamforming}, the vectors $\mathbf{f}[i]$ are optimized for each subcarrier index $i\in\{1,\dots ,I\}$ without considering the key equation \eqref{def:f} which establishes the connection between $\mathbf{f}[i]$ and the \ac{sim}'s propagation matrix $\mathbf{G}[i]$ (defined in \eqref{eq:sim_end_to_end}). 
In the rest of this section, we denote these optimized vectors as $\hat{\mathbf{f}}[i]$ (i.e., the outputs of Algorithm~\ref{alg:alt_opt}). 
Observing the signal models in \eqref{eq:SUE-received-sig} and \eqref{def:los_h_su_s}, to optimize the \ac{sim}'s phase coefficients, we try to choose the phase shifts to achieve the most similar beampatterns obtained by $\{\hat{\mathbf{f}}[i]\}_{i=1}^{I}$. 
Thus, the problem can be formulated as follows:
\begin{equation}\label{opt:sim_phase}
\min_{\{\boldsymbol{\Phi}^{(\ell)}\}_{l=1}^{L}\ 
}\ 
\frac{1}{I}
\sum_{i=1}^{I}
\lVert
\mathbf{b}[i]
-
\mathbf{q}[i]
\rVert^2,\ 
\mathrm{s.t.}\ 
\mathbf{b}[i]
=
\mathbf{A}[i]^{\mathsf{T}}
\mathbf{G}[i]
\mathbf{w}[i],
\end{equation}
where 
$\mathbf{q}[i]=\mathbf{A}[i]^{\mathsf{T}}\hat{\mathbf{f}}[i]$ 
denotes the beampattern of the vector $\hat{\mathbf{f}}[i]$. 
$\mathbf{A}[i]=[\mathbf{a}_i(\boldsymbol{\theta}_{1}),\dots ,\mathbf{a}_i(\boldsymbol{\theta}_{N_g})]$, is the collection of \ac{arv} at $N_g$ 2D-\ac{aod} grids in the $i$-th subcarrier. 
To solve \eqref{opt:sim_phase}, and motivated by the structural similarity between \eqref{eq:sim_end_to_end} and \acp{dnn}, we adopt a mini-batch optimization approach. 
In each batch we evaluate the beampattern-matching loss \eqref{opt:sim_phase} on a randomly drawn angular grid of size $N_g$ and update the \ac{sim} phase coefficients via backpropagation and a gradient-based optimizer. 
Implementation details are given in the following subsections.

\subsection{Forward and backward recursions}
To evaluate $\mathbf{f}[i]$ and to compute gradients efficiently we use right (forward) and left (backward) recursions.

Define right vectors $\mathbf{r}^{(\ell)}[i]\in\mathbb{C}^N$ for layer $\ell$ and subcarrier $i$ by
$\mathbf{r}^{(1)}[i]
= 
\mathbf{w}[i],$
and for $\ell=2,\dots,L$,
\begin{equation}
\mathbf{r}^{(\ell)}[i] = \mathbf{W}^{(\ell)}[i] \big( \boldsymbol{\Phi}^{(\ell-1)} \mathbf{r}^{(\ell-1)}[i] \big).
\end{equation}
Thus the vector entering layer $\ell$ (before multiplication by the diagonal phase $\boldsymbol{\Phi}^{(\ell)}$) is $\mathbf{r}^{(\ell)}[i]$.
In a similar manner, we define left matrices $\mathbf{L}^{(\ell)}[i]\in\mathbb{C}^{N\times N}$ by
$\mathbf{L}^{(L)}[i] = \mathbf{I}_N,$
and for $\ell=L-1,\dots,1$,
\begin{equation}
\mathbf{L}^{(\ell)}[i] 
= 
\mathbf{L}^{(\ell+1)}[i] 
\big( 
\boldsymbol{\Phi}^{(\ell+1)} \mathbf{W}^{(\ell+1)}[i] \big).
\end{equation}

With these definitions, for $\ell\in\{1\,\dots, L\}$, the \ac{sim} output response vector $\mathbf{f}[i]$ in \eqref{def:f} can be represented as:
\begin{equation}
\mathbf{f}[i] = \mathbf{L}^{(\ell)}[i]\big( \boldsymbol{\Phi}^{(\ell)} 
\mathbf{r}^{(\ell)}[i] \big).
\label{eq:f_via_LR}
\end{equation}

\subsection{Per-element gradient derivation}
For a single subcarrier $i$ and a given angular batch (we omit the batch index for clarity), define the residual vector
\begin{equation}
\mathbf{r}[i]
= 
\mathbf{b}[i] - \mathbf{q}[i].
\end{equation}
The per-subcarrier contribution to the loss is $J[i] 
= 
\mathbf{r}[i]^{\mathsf{H}} \mathbf{r}[i]$. 
By chain rule,
\begin{equation}
\frac{\partial J[i]}{\partial \phi_n^{(\ell)}}
= 
2\,\Re\!\Big\{ \Big( \frac{\partial \mathbf{b}[i]}{\partial \phi_n^{(\ell)}} \Big)^{\mathsf{H}} 
\mathbf{r}[i] 
\Big\}.
\end{equation}
From \eqref{eq:f_via_LR} only the $n$-th diagonal element of $\boldsymbol{\Phi}^{(\ell)}$ affects $\phi_n^{(\ell)}$, hence:
\begin{equation}\label{eq:partial_b_phi}
\frac{\partial \mathbf{b}[i]}{\partial \phi_n^{(\ell)}}
= j\,\mathbf{A}[i]^{\mathsf{T}}\,\mathbf{L}^{(\ell)}[i]\,\mathbf{e}_n \,e^{j\phi_n^{(\ell)}} r_{n}^{(\ell)}[i],
\end{equation}
where $\mathbf{e}_n$ is the $n$-th column of $\mathbf{I}_N$ and $r_{n}^{(\ell)}[i]$ is the $n$-th entry of $\mathbf{r}^{(\ell)}[i]$.
The backprojected vector $\mathbf{u}^{(\ell)}[i]\in\mathbb{C}^{N}$ for layer $\ell$ is defined as:
\begin{equation}
\mathbf{u}^{(\ell)}[i] 
= 
\mathbf{L}^{(\ell)}[i]^{\mathsf{H}} 
\mathbf{A}[i]^{*}
\,\mathbf{r}[i] 
\in\mathbb{C}^N,
\end{equation}
and denote its $n$-th entry by $u_{n}^{(\ell)}[i]$. 
Substituting this into \eqref{eq:partial_b_phi} and simplifying:
\begin{align}
\frac{\partial J[i]}{\partial \phi_n^{(\ell)}}
&= 
2\Re\!\Big\{ -j\,e^{-j\phi_n^{(\ell)}} r_{i,n}^{(\ell)*}\, u_{i,n}^{(\ell)} \Big\} \nonumber\\
&= 2\,\Im\!\big\{ e^{-j\phi_n^{(\ell)}}\, r_{i,n}^{(\ell)*}\, u_{i,n}^{(\ell)} \big\}.
\label{eq:grad_element}
\end{align}
By stacking the gradients across $n=1,\dots,N$, the gradient vector $\mathbf{g}^{(\ell)}[i]\in \mathbb{R}^N$ for layer $\ell$ can be represented as:
\begin{equation}
\mathbf{g}^{(\ell)}[i]
=
\nabla_{\bm{\Phi}^{(\ell)}} J[i]
= 2\,\Im\!\big\{ \boldsymbol{\Phi}^{(\ell)^*} 
(
\mathbf{r}^{(\ell)^*}[i] 
\odot 
\mathbf{u}^{(\ell)}[i] 
)
\big\}.
\label{eq:grad_layer}
\end{equation}

\subsection{Mini-batch gradient and update rule}
When using a mini-batch of $N_g$ angular samples, the average batch loss over subcarriers can be defined as:
\begin{equation}
\mathcal{L}_{\text{batch}} = \frac{1}{I}\sum_{i=1}^{I} 
J[i],
\end{equation}
whose gradient can be obtained as:
\begin{equation}
\nabla_{\bm{\Phi}^{(\ell)}} \mathcal{L}_{\text{batch}}
= \frac{1}{I}\sum_{i=1}^{I} \nabla_{\bm{\Phi}^{(\ell)}} 
J[i].
\end{equation}

We use the Adam optimizer to update the \ac{sim} phase coefficients whose details are provided in Algorithm~\ref{alg:phase_opt}.

\begin{algorithm}[ht]
\caption{SIM Coefficients Optimization}\label{alg:phase_opt}
\begin{algorithmic}[1]
\State 
\textbf{Inputs:} 
Initial \ac{sim} phase coefficients $\{\bm{\Phi}^{(\ell)}\}$, 
transmission matrices $\{\mathbf{W}^{(\ell)}[i],\mathbf{w}[i]\}$, hyperparameters: learning rate $\eta$, Adam parameters $(\beta_1,\beta_2,\varepsilon)$, batch size $N_g$, number of batches per epoch $N_b$, number of epochs $N_e$.
\State 
$t\!\leftarrow\!0$, $\mathbf m^{(\ell)}\!\leftarrow\!0$, $\mathbf v^{(\ell)}\!\leftarrow\!0$ for $\ell=1\!:\!L$.
\For{each epoch}
  \For{each mini-batch}
    \State 
    \multiline{%
    Sample $N_g$ angular grid points for this batch and form $\mathbf{A}[i]$ on the batch and obtain $\mathbf q[i]=\mathbf{A}[i]^{\mathsf{T}}\hat{\mathbf{f}}[i]$.
    }
    \For{each subcarrier $i=1,\dots,I$}
      \State Compute right recursion $\mathbf{r}^{(\ell)}[i]$ for $\ell=1: L$
      \State Compute left matrices $\mathbf{L}^{(\ell)}[i]$ for $\ell=L: 1$
      \State Compute $\mathbf{f}[i]$ and $\mathbf{b}[i] = \mathbf{A}[i]^{\mathsf{T}}\mathbf{f}[i]$.
      \State Compute residual $\mathbf{r}[i] = \mathbf{b}[i] - \mathbf{q}[i]$.
      \State 
      \multiline{%
      Compute $\mathbf{u}^{(\ell)}[i] = \mathbf{L}^{(\ell)^\mathsf{H}}[i]\mathbf{A}[i]^{*} \mathbf{r}[i]$.
      }
      \State 
      \multiline{%
      Obtain the gradient vectors $\mathbf{g}^{(\ell)}[i]$ using \eqref{eq:grad_layer} for $\ell=1: L$.
      }
    \EndFor
    \State 
    \multiline{%
    Compute the average batch gradient $\bar{\mathbf{g}}^{(\ell)}=\frac{1}{IN_g}\sum_{i=1}^{I}\mathbf{g}^{(\ell)}[i]$.
    }
    \State 
    {\bf Adam moment updates:}
    \begin{align*}
    &t\gets t+1,\\
    &\mathbf m^{(\ell)} \leftarrow \beta_1 \mathbf m^{(\ell)} + (1-\beta_1)\,\bar{\mathbf{g}}^{(\ell)},\\
    & \mathbf v^{(\ell)} \leftarrow \beta_2 \mathbf v^{(\ell)} + (1-\beta_2)\,\big(\bar{\mathbf{g}}^{(\ell)}\odot\bar{\mathbf{g}}^{(\ell)}\big),\\
    &\widehat{\mathbf m}^{(\ell)}=\frac{\mathbf m^{(\ell)}}{1-\beta_1^{t}},\qquad
      \widehat{\mathbf v}^{(\ell)}=\frac{\mathbf v^{(\ell)}}{1-\beta_2^{t}}.
    \end{align*}
    \State 
    \multiline{%
    {\bf Phase coefficients update:}
    
    $
      \bm{\Phi}^{(\ell)} \leftarrow \operatorname{wrap}\!\Big(\bm{\Phi}^{(\ell)} - \eta\;\frac{\widehat{\mathbf m}^{(\ell)}}{\sqrt{\widehat{\mathbf v}^{(\ell)}}+\varepsilon}\Big),
    $
    where $\operatorname{wrap}(\cdot)$ projects angles into $(-\pi,\pi]$.
    }
  \EndFor
\EndFor
\end{algorithmic}
\end{algorithm}

\section{Complexity Analysis}
\subsection{Complexity of Algorithms~\ref{alg:F_opt} and \ref{alg:alt_opt}}
The dominant computational burden in Algorithm~\ref{alg:F_opt} is the repeated computation of the principal eigenvector of the shifted matrix 
$\mathbf{A}-\mu\mathbf{R}$ during the bisection/bracketing phase. 
Let $M_1$ denote the number of bisection iterations needed to reach the tolerance $\xi_{\mathrm{tol}}$ and by $T_{\mathrm{eig}}(N)$ the cost of a single principal-eigenpair computation, the per-subcarrier complexity is $O\left(M_1\,T_{\mathrm{eig}}(N)\right)$. 
Thus, for $I$ subcarriers the overall complexity scales as $O\left(I\,M_1\,T_{\mathrm{eig}}(N)\right)$. 
Note that a full eigenvalue decomposition incurs complexity $T_{\mathrm{eig}}=O(N^3)$. 
However, since only the principal eigenvector is required, much more efficient iterative methods (e.g., Lanczos) can be employed using a few matrix-vector products per iteration (roughly $O(N^2)$ per iteration for dense matrices) with fast convergence \cite[Chap.~10]{golub2013matrix}.
Regarding the outer optimization, since $\mathbf{J}_B\in \mathbb{R}^{5\times 5}$, the closed form solution \eqref{eq:d_star_5} of $\{\mathbf{d}_j\}_{j=1}^{3}$ has the complexity $O(1)$. 
Hence, the complexity of Algorithm~\ref{alg:alt_opt}, can be obtained as $O\left(I\,M_2\,M_1\,T_{\mathrm{eig}}(N)\right)$
where $M_2$ is the average number of the alternating iterations until convergence.
In Sec.~\ref{sec:simul_converge}, Monte-Carlo simulations are conducted to demonstrate the convergence behavior of these algorithms.
\vspace{-0.5cm}
\subsection{Complexity of Algorithm~\ref{alg:phase_opt}}
The dominant cost inside each mini-batch arises from constructing the backward left matrices and performing dense matrix-matrix and matrix-vector products for each subcarrier. 
Conservatively, the per-mini-batch complexity scales as $O\big(I(LN^3 + N_g N + L N^2)\big)$, so that the total training complexity is
$
O\!\big(N_e N_b I (L N^3 + N_g N + L N^2)\big)
$, 
and in the large-$N$ regime is dominated by $O(N_e N_b I L N^3)$. 
It is noteworthy that, despite the obtained asymptotic arithmetic cost, parallelizing across subcarriers and batches and using optimized GPU kernels, substantially reduces this cost. 
\begin{remark}
While we assume quasi-static channels in this paper, Algorithm~\ref{alg:phase_opt} can be adapted to real-time scenarios by initializing with the previous optimized phase coefficients and performing limited online fine-tuning (small batches and learning rates), or by using precomputed phase-codebooks\cite{An2024Codebook, fadakar2025mutual, fadakar2025hybrid}. 
The details are left for future work.
\end{remark}

\vspace{-0.5cm}
\section{Simulations}
In this section, we evaluate the performance of the proposed methods for \ac{sim}-assisted \ac{mmwave} systems through extensive numerical simulations.
\vspace{-0.5cm}
\subsection{Simulation Setup}\label{sec:simulation_setup}
The \ac{sim} layers are assumed to be equipped with a \ac{upa} with the same number of rows and columns $N^h=N^v$. 
For the \ac{su}'s state parameter vector's prior $p_{\boldsymbol{\gamma}}$, we assume a uniform distribution inside the cuboid specified by $50\le x\le 70$, $-10\le y\le 10$, and $0\le z\le 5$.
The number of $M_p=20000$ points are generated based on the prior $p_{\boldsymbol{\gamma}}$ to estimate the expectation $\mathbb{E}_{\boldsymbol{\gamma}}$ in \eqref{eq:FIM_entry}. 
The positions of the components of the considered system as well as the uncertainty region for the \ac{su} are shown in Fig.~\ref{fig:Env_map}, 
the default system and algorithmic parameters are presented in Table~\ref{tab:sys-params}, and Table~\ref{tab:train-params}, respectively. 
In all simulations except Sec.~\ref{sec:robustness_Npu}, we use $N_\text{pu}=2$ \acp{pu} (i.e., \ac{pu} 1 and \ac{pu} 2 depicted in Fig.~\ref{fig:Env_map}). 
For the \ac{sim}'s geometry parameters including the area of each meta-atom and inter-layer distance, trial and error and also inspiration from the existing literature (e.g., \cite{An2024SIM_DOA, li2025stacked_ofdm}) were used to obtain suitable parameters which are shown in Table~\ref{tab:sys-params}.
Moreover, the number of $Q_{\text{su,s}}=Q_{\text{pu,s}}=Q_{\text{su,pb}}=Q_{\text{pu,pb}}=50$ single-point scatterers are randomly generated in the environment to model \acp{mpc} for the geometric channels in \eqref{def:channels}. 

It is noteworthy that according to \eqref{eq:Sommerfeld_formula}, due to different amplitudes in transmission matrices $\mathbf{W}^{(\ell)}[i]$ between layers, the gain of the end-to-end \ac{sim} response in \eqref{eq:sim_end_to_end} or \eqref{def:f}, may vary for different numbers of \ac{sim} meta-atoms $N$ and the number of layers $L$. 
Thus, for a fair comparison in different simulations, we define the total transmitted power from the \ac{sws} i.e., the total radiated power from the \ac{sim}'s output layer as:
\begin{equation}\label{def:p_sws}
P_{\text{sws}}
=
\frac{IP_\text{sb}}{\sum_{i=1}^{I}
\lVert\mathbf{f}[i]
\rVert^2},
\end{equation}
which is set to $P_{\text{sws}}=30\,\mathrm{dBm}$ as a default value, and the $\ac{pb}$ power is fixed at $P_\text{pb}=30\,\mathrm{dBm}$ by default. 
It is noteworthy that by interpreting the \ac{sb} omni-antenna and \ac{sim} as an adaptive antenna, $P_{\text{sws}}$ corresponds to the the product of conducted power and antenna efficiency. 
Under these power levels, we choose the \ac{qos} parameter $\kappa=98\%$ in \eqref{const:pu_interference}.
Note that some parameters in this subsection may differ across different simulation scenarios.

\begin{figure}[!t]
\centering
\includegraphics[width=\columnwidth]{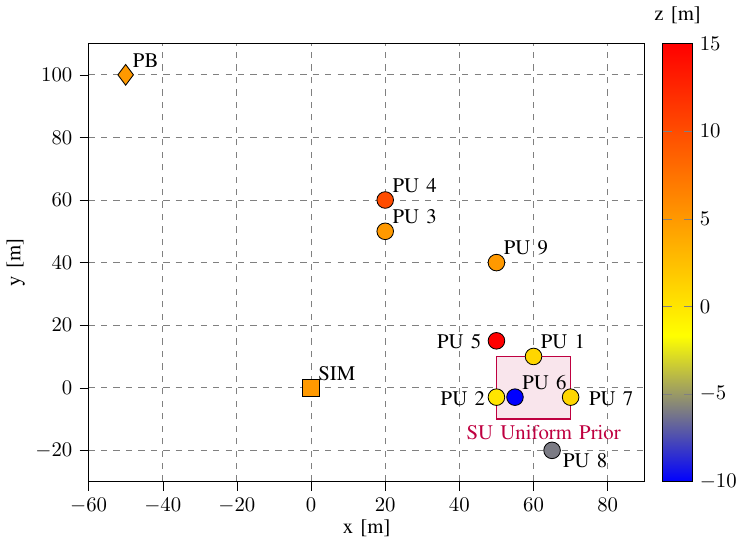}
\caption{
Positions of the \ac{sim}, \ac{pb}, \acp{pu}, and the \ac{su} uncertainty region considered in the simulations based on the selected uniform prior. 
The colorbar shows the corresponding $z$ value coordinates.
}
\label{fig:Env_map}
\end{figure}

\begin{table}[ht]
\caption{\label{tab:sys-params} System parameters.}
\centering
\fontsize{12}{10}\selectfont 
\resizebox{\columnwidth}{!}{
\begin{tabular}{ |l|l|  }
\hline
Default System Parameters and
Symbol
&
\textbf{Value}
\\
\hline
Light speed $c$ & $3\times 10^8\,\mathrm{m/s}$
\\
Carrier frequency $f_c$ and wavelength $\lambda_c=c/f_c$ & $30\,\mathrm{GHz},\ 1\,\mathrm{cm}$
\\
Noise PSD $N_0$ & $-173.855\,\mathrm{dBm}$
\\
Underlay bandwidth $B$ & $50\,\mathrm{MHz}$
\\
Underlay subcarrier spacing $\Delta f$ & $1\,\mathrm{MHz}$
\\
\ac{sws} transmit power $P_\text{sws}$ & $30\,\mathrm{dBm}$
\\
\ac{pb} transmit power $P_\text{pb}$ & $30\,\mathrm{dBm}$
\\
\ac{qos} parameter $\kappa$ in \eqref{const:pu_interference} & $98\%$
\\
Frequency of the $i$-th subcarrier
$f_i$ &  $f_c-(i-1)\Delta f$
\\
\ac{sb} position $\mathbf{p}_\text{sb}$ & $[0,0,5]^\mathsf{T}$
\\
\ac{pb} position $\mathbf{p}_\text{pb}$ & $[-50,100,5]^\mathsf{T}$
\\
Number of \ac{sim} layers 
$L$ &  $4$
\\
Inter-layer distance in the \ac{sim} 
$d_s$ &  $3\lambda_c/2$
\\
Distance between two adjacent meta-atoms
$d$ &  $\lambda_c/2$
\\
Area of each meta-atom
$A_s$ &  $\lambda_c/2$
\\
\ac{sim} per-layer number of rows and columns $N^{h}$, $N^v$ & $6$, $6$
\\
Number of generated points $M_p$ to estimate $\mathbb{E}_{\boldsymbol{\gamma}}$ & $20000$
\\
Upper bounds $\delta[i]$ in \eqref{const:F_pow} & $1$
\\
\hline
\end{tabular}
}
\end{table}

\begin{table}[ht]
\caption{\label{tab:train-params} Algorithmic parameters.}
\centering
\fontsize{12}{10}\selectfont 
\resizebox{\columnwidth}{!}{
\begin{tabular}{ |l|l|  }
\hline
Parameter
&
\textbf{Value}
\\
\hline
Parameter $\xi_\mathrm{tol}$ in Algorithm~\ref{alg:F_opt} & $10^{-20}$
\\
Parameters $\varepsilon_{\text{tol}},\tau_{\text{tol}}$ in Algorithm~\ref{alg:alt_opt} 
& $10^{-12}$, $10^{-12}$
\\
Number of epochs $N_e$ & $200$
\\
Number of batches per epoch & $50$
\\
Batch size $N_d$ & $512$
\\
Learning rate
$\eta$ &  $0.001$
\\
First moment exponential decay rate $\beta_1$ & $0.9$
\\
Second moment exponential decay rate $\beta_2$ & $0.999$
\\
Numerical stability constant $\varepsilon$ & $10^{-8}$
\\
\hline
\end{tabular}
}
\end{table}
\vspace{-0.5cm}
\subsection{Convergence Analysis of Algorithms~\ref{alg:F_opt} and \ref{alg:alt_opt}}\label{sec:simul_converge}
To assess the performance of the proposed \ac{ao}-based Algorithm~\ref{alg:alt_opt}, Fig.~\ref{fig:alternate_bcrb_iter} illustrates the \ac{bcrb} versus number of iterations for \ac{sim} layer sizes $N^h=N^v\in\{5,6,7,8\}$. 
The results are averaged over $200$ Monte-Carlo trials. 
According to the results, the algorithm converges rapidly, within approximately $4$ iterations. 
\begin{remark}
In the simulations that follow, the optimized end-to-end \ac{sim} response vectors $\{\hat{\mathbf{f}}[i]\}_{i=1}^{I}$ produced by Algorithm~\ref{alg:alt_opt} are hereafter referred to as the optimal estimates of $\{\mathbf{f}[i]\}_{i=1}^{I}$. 
As detailed in Sec.~\ref{sec:sim_opt}, the \ac{sim} coefficients are trained to reproduce the beampatterns corresponding to these optimal vectors, i.e., to match the target beampatterns as closely as possible.

\end{remark}

\begin{figure}[!t]
\centering
\includegraphics[width=\columnwidth]{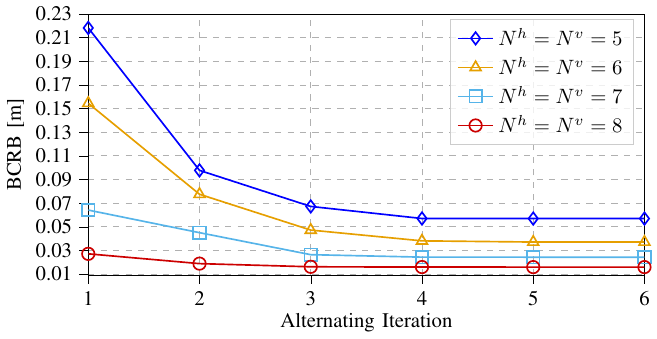}
\caption{
performance of the \ac{ao} based Algorithm~\ref{alg:alt_opt}.
}
\label{fig:alternate_bcrb_iter}
\end{figure}

To evaluate the performance of Algorithm~\ref{alg:F_opt} we initialize the vectors $\{\mathbf{f}[i]\}_{i=1}^{I}$ randomly, and fix the vectors $\{\mathbf{d}_j\}_{j=1}^{3}$ using the optimal values in the last iteration of Algorithm~\ref{alg:alt_opt}. 
Then, we execute Algorithm~\ref{alg:F_opt} to optimize $\{\mathbf{f}[i]\}_{i=1}^{I}$.
Fig.~\ref{fig:optF_bcrb_hist_semilogy} shows the \ac{bcrb} after sequentially optimizing each $\mathbf{f}[i]$, $i=1,\dots,I$, with Algorithm~\ref{alg:F_opt} for \ac{sim} layer sizes $N^h=N^v\in\{5,6,7,8\}$. 
The curves evidently decrease with the subcarrier index, which demonstrates the convergence and effectiveness of the proposed algorithm. 
As stated in Prop.~\ref{prop:inner_opt}, since each subcarrier can be optimized independently, the subcarrier update order does not affect the final \ac{bcrb} i.e., different ordering choices converge to the same solution. 

\begin{figure}[!t]
\centering
\includegraphics[width=\columnwidth]{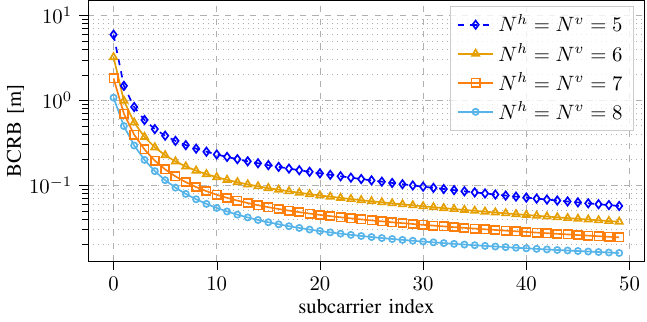}
\caption{
\ac{bcrb} values after optimizing $\mathbf{f}[i]$ for $i=1,\dots ,I$.
}
\label{fig:optF_bcrb_hist_semilogy}
\end{figure}

Fig.~\ref{fig:optF_r_energy_traces_semilogy} shows the relative objective $g(\mu)/\epsilon = \mathbf{v}(\mu)^{\mathsf{H}} \mathbf{R} \mathbf{v}(\mu)/\epsilon$ (i.e., relative interference energy) produced by the bisection subroutine in Algorithm~\ref{alg:F_opt}, plotted against the iteration index. 
Individual per-subcarrier objective values are plotted together with their mean. 
The results indicate that the proposed bisection routine converges to the optimal value $1$ within approximately $10$ iterations for every subcarrier, demonstrating the efficiency and robustness of the proposed method to satisfy the interference constraint \eqref{const:pu_interference} (or equivalent constraint \eqref{const:pu_rate}) for inner optimization of \eqref{opt:equivalent_v1}.

\begin{figure}[!t]
\centering
\includegraphics[width=\columnwidth]{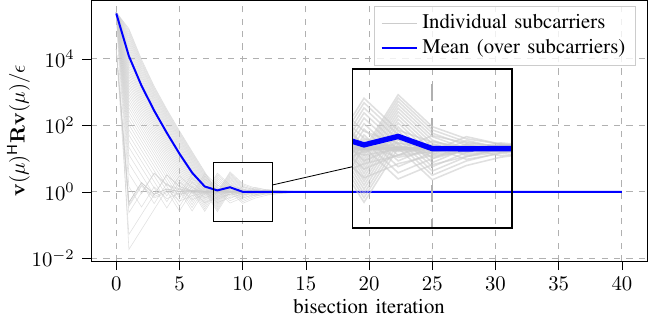}
\caption{
Relative interference energy per bisection iteration for different subcarriers. 
The average curve over all subcarriers is also shown. 
}
\label{fig:optF_r_energy_traces_semilogy}
\end{figure}

\vspace{-0.2cm}
\subsection{Training SIM and Convergence Analysis}
Algorithm~\ref{alg:phase_opt} is employed to train the \ac{sim} phase coefficients. 
Training is performed for $N_e=200$ epochs with $N_b=50$ batches per epoch, each batch containing $N_d=512$ samples. 
Detailed training parameters are given in Table~\ref{tab:train-params}. 
Let $\bm{\nabla}_{e,b}\in\mathbb{C}^{L\times N}$ denote the gradient computed at epoch $e$ and batch $b$, and define the epoch-averaged gradient norm $\bar{g}_e = \frac{1}{N_b}\sum_{b=1}^{N_b}\lVert \bm{\nabla}_{e,b}\rVert_F$. Fig.~\ref{fig:grad_norm_epoch} plots $\bar{g}_e$ versus epoch $e$ for $N^h=N^v\in\{6,8\}$ and $L\in\{2,4\}$, which demonstrate convergence of the training. 
The corresponding normalized beampattern difference between the \ac{sim} and the optimal beampattern are shown in Fig.~\ref{fig:bp_diff_hist}, further confirming convergence of the training process.

\begin{figure}
\centering
\begin{subfigure}{\columnwidth}
\centering
\includegraphics[width=\columnwidth]{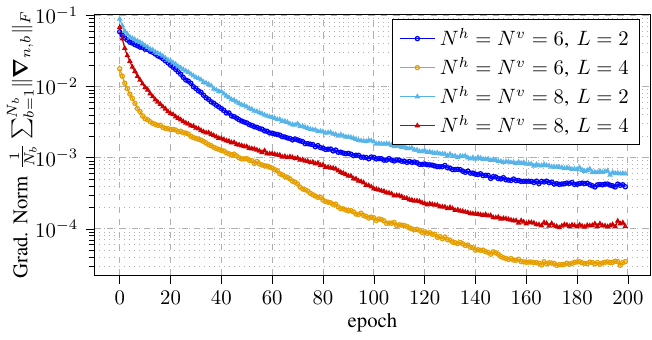}
\caption{Average gradient norms versus epoch to assess convergence behavior.}
\label{fig:grad_norm_epoch}
\end{subfigure}%
\hfill
\vspace{0pt}
\begin{subfigure}{\columnwidth}
\centering
\includegraphics[width=\columnwidth]{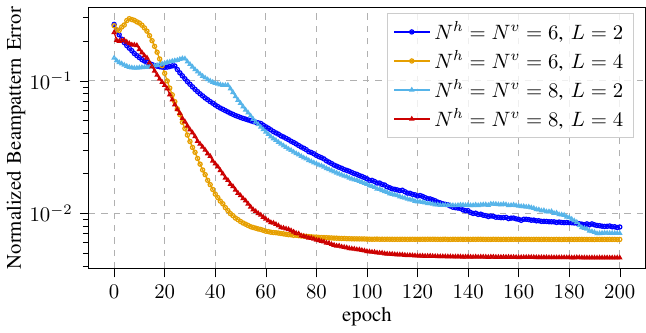}
\caption{Normalized beampattern error during training versus epoch.}
\label{fig:bp_diff_hist}
\end{subfigure}%
\caption{
Convergence analysis of the proposed Algorithm~\ref{alg:phase_opt}.
}
\label{fig:conv_analysis}
\end{figure}

\begin{figure*}[ht]
\centering
\includegraphics[width=2\columnwidth]{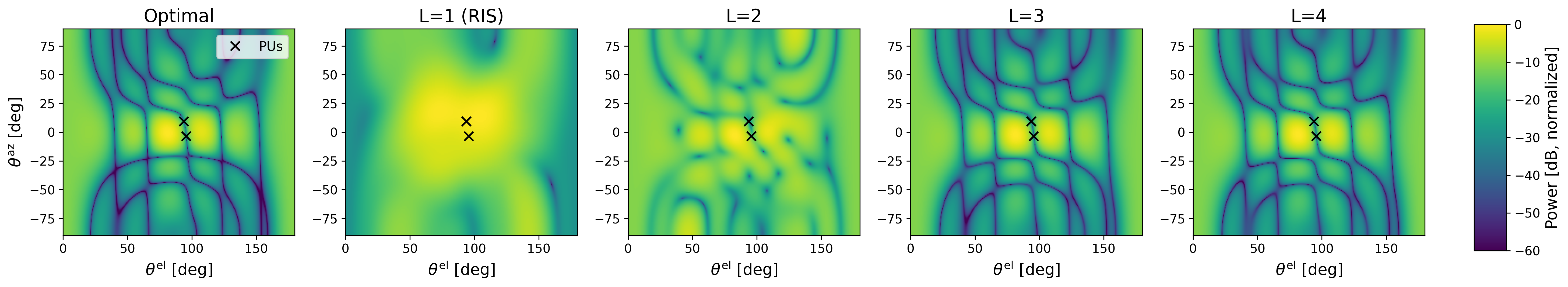}
\caption{
2D-Beampatterns of the trained \acp{sim} with fixed $N^h=N^v=6$ and different number of layers $L\in\{1,2,3,4\}$.
}
\label{fig:beampatterns_2d}
\end{figure*}

\begin{figure}
\centering
\begin{subfigure}{\columnwidth}
\centering
\includegraphics[width=\columnwidth]{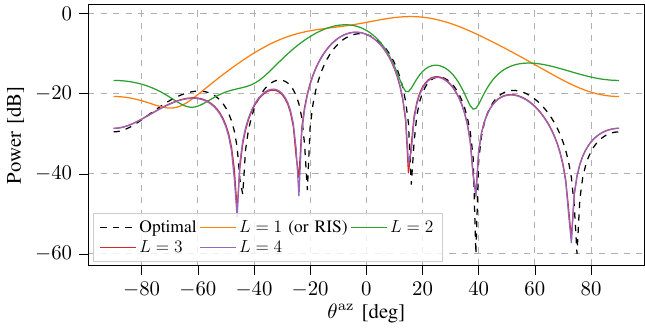}
\caption{1D beampattern comparison for fixed $\theta^{\text{el}}=90^\circ$}
\label{fig:beampatterns_1d_phi}
\end{subfigure}%
\hfill
\vspace{6pt}
\begin{subfigure}{\columnwidth}
\centering
\includegraphics[width=\columnwidth]{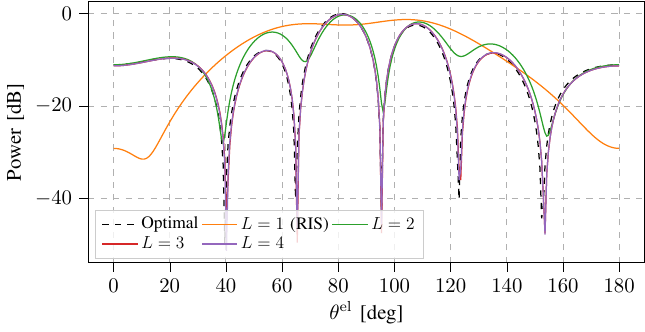}
\caption{1D beampattern comparison for fixed $\theta^{\text{az}}=0^\circ$}
\label{fig:beampatterns_1d_theta}
\end{subfigure}%
\caption{
Two 1D cuts of the 2D beampatterns to better compare the performance of the trained \ac{sim} for different number of layers.
}
\label{fig:1D-patterns}
\end{figure}
\vspace{-0.5cm}
\subsection{Beampattern Analysis}
Fig.~\ref{fig:beampatterns_2d} presents examples of the normalized 2D beampatterns for the first subcarrier  comparing the analytically obtained optimal response and the trained \ac{sim} for layer counts $L\in\{1,2,3,4\}$. 
Note that $L=1$ corresponds to a conventional \ac{ris} operating in diffraction mode. The 2D-\acp{aod} of the \acp{pu} are indicated in each subplot. 
For $L\in\{1,2\}$ the trained \ac{sim} shows a large mismatch to the optimal beampattern and does not sufficiently attenuate power at the \acp{pu}. 
By contrast, when $L\in\{3,4\}$ the trained \ac{sim} closely reproduces the optimal beampattern and the \acp{pu} 2D-\acp{aod} lie at the beampattern nulls. 
To facilitate comparison, in Fig.~\ref{fig:1D-patterns}, we plot two 1D cuts of the 2D beampatterns from Fig.~\ref{fig:beampatterns_2d} taken at fixed angles $ \theta^{\text{el}}=90^\circ$ and $ \theta^{\text{az}}=0^\circ$. 
The plots confirm that the \acp{sim} with $L\in\{3,4\}$ closely reproduce the optimal beampattern.

\vspace{-0.1cm}
\subsection{Localization and Spectral Efficiency Performance}
In this subsection, we evaluate the localization and \ac{se} performance of the proposed methods. 
Fig.~\ref{fig:bcrb_vs_P_dBm} shows the \ac{bcrb} obtained with trained \acp{sim} having $L=4$ layers versus $P_{\text{sws}}$ (defined in \eqref{def:p_sws}).
It is observed that the learned responses closely track the corresponding optimal values, demonstrating the performance of Algorithm~\ref{alg:phase_opt}. 
To quantify the effect of the number of layers, Fig.~\ref{fig:BCRB_vs_L} compares the \ac{bcrb} for trained \acp{sim} with $N^h=N^v\in\{6,8\}$ and $L\in\{1,2,3,4\}$ at a fixed power $P_{\text{sws}}=20\,\mathrm{dBm}$, showing that localization performance improves with $L$ and essentially converges to the optimal end-to-end response by $L=4$. 
Hence, for the considered system, a four-layer \ac{sim} is sufficient to achieve optimal localization.

\begin{figure}[ht]
\centering
\includegraphics[width=\columnwidth]{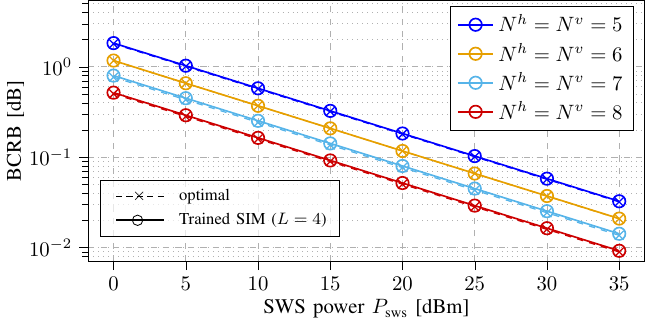}
\caption{
\ac{bcrb} values versus $P_{\text{sws}}$ using trained \acp{sim} with $L=4$ layers.
}
\label{fig:bcrb_vs_P_dBm}
\end{figure}

\begin{figure}[ht]
\centering
\includegraphics[width=\columnwidth]{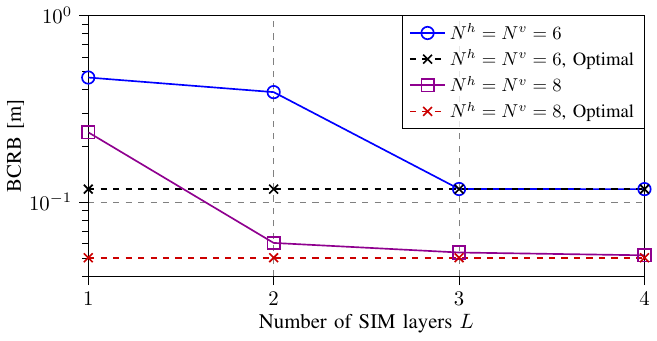}
\caption{
\ac{bcrb} of the trained \acp{sim} versus $L$ at $P_{\text{sws}}=20\,\mathrm{dBm}$.
}
\label{fig:BCRB_vs_L}
\end{figure}

\begin{figure}[ht]
\centering
\includegraphics[width=\columnwidth]{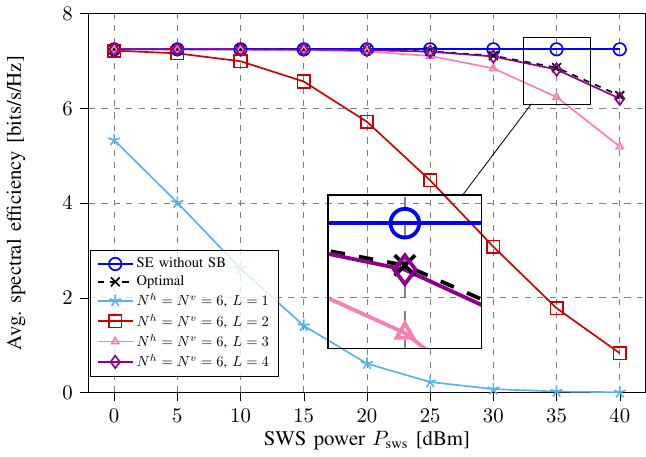}
\caption{
Average \ac{se} versus $P_{\text{sws}}$ at a fixed $P_\text{pb}=30\,\mathrm{dBm}$.
}
\label{fig:SE_vs_P_sws}
\end{figure}

The average \ac{se}, defined as $ \frac{1}{I}\sum_{i=1}^{I}\log\big(1+\mathrm{SINR}[i]\big)$ (with $\mathrm{SINR}[i]$ given in \eqref{eq:avg_SINR}), is plotted in Fig.~\ref{fig:SE_vs_P_sws} versus $P_{\text{sws}}$ for trained \acp{sim} with $N^h=N^v=6$ and $L\in\{1,2,3,4\}$ at a fixed \ac{pb} transmit power $P_{\text{pb}}=30\,\mathrm{dBm}$ where the optimal \ac{se} curve is also shown as a baseline. 
Evidently, the trained \ac{sim} with $L=4$ closely achieves the optimal performance. 

Likewise, the corresponding \ac{se} are shown in terms of \ac{pb} power $P_\text{pb}$ in Fig.~\ref{fig:SE_vs_PBS}. 
The \ac{sim} with $L=4$ achieves near-optimal \ac{se} across the considered power range, and \acp{sim} with $L\ge 2$ substantially outperform the single-layer \ac{ris}. 
These results demonstrate that multi-layer \acp{sim} are capable of learning and realizing the sophisticated multi-functional beampatterns required to jointly support \ac{su} localization and \ac{pu} interference management.

\begin{figure}[ht]
\centering
\includegraphics[width=\columnwidth]{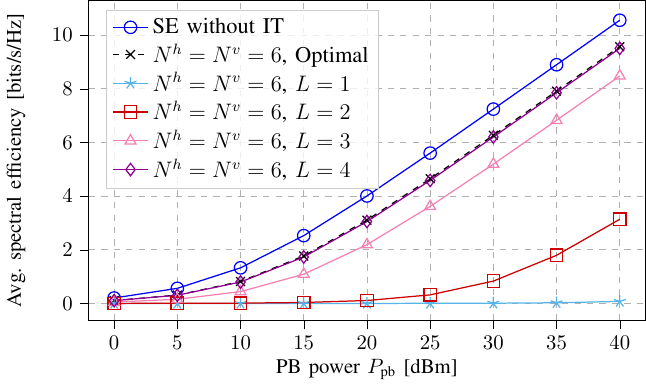}
\caption{
Average \ac{se} versus $P_{\text{pb}}$ at a fixed $P_\text{sws}=40\,\mathrm{dBm}$.
}
\label{fig:SE_vs_PBS}
\end{figure}
\vspace{-0.4cm}
\subsection{Robustness Under Different Number of PUs}\label{sec:robustness_Npu}
In this subsection we evaluate the proposed methods as the number of active \acp{pu} varies. 
We consider a set of nine candidate \acp{pu} whose locations are shown in Fig.~\ref{fig:Env_map}.
For each $N_{\text{pu}}\in\{1,3,5,7,9\}$ the first $N_{\text{pu}}$ positions are treated as active and the \ac{sim} is trained for that scenario. 
Fig.~\ref{fig:Num_PU_beampatterns_2d} presents the resulting 2D beampatterns with the corresponding 2D-\acp{aod} of the active \acp{pu} indicated. 
In every case the trained \ac{sim} places deep nulls at the active-\ac{pu} directions to satisfy the interference constraint, yielding distinct beampatterns as $N_{\text{pu}}$ changes. 
The corresponding \ac{bcrb} values are illustrated in Fig.~\ref{fig:bcrb_vs_Npu}.
As expected, increasing $N_{\text{pu}}$ generally increases the \ac{bcrb} (i.e., degrades localization accuracy) because more stringent \ac{pu} protection reduces the degrees of freedom available for localization. 
The magnitude of this degradation, however, is geometry-dependent: when active \acp{pu} lie close to or inside the \ac{su} uncertainty region (e.g., \acp{pu} 1, 2, 6, and 7 in Fig.~\ref{fig:Env_map}), the performance penalty is larger.

\begin{figure*}[ht]
\centering
\includegraphics[width=2\columnwidth]{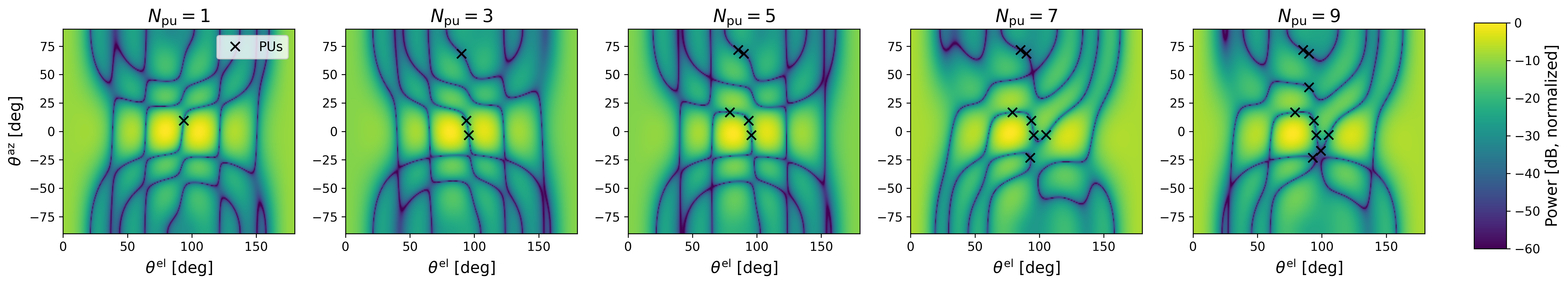}
\caption{
2D-Beampatterns of the trained \acp{sim} (with $L=4$ layers and $N^h=N^v=6$) for different number of active \acp{pu}. 
}
\label{fig:Num_PU_beampatterns_2d}
\end{figure*}

\begin{figure}[!t]
\centering
\includegraphics[width=\columnwidth]{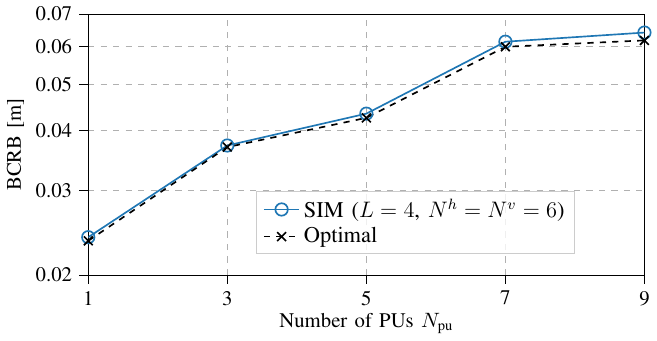}
\caption{
\ac{bcrb} versus number of \acp{pu} at $P_\text{sws}=30\,\mathrm{dBm}$.
}
\label{fig:bcrb_vs_Npu}
\end{figure}
\vspace{-0.4cm}
\section{Conclusion}
This paper investigated a multi-functional \ac{sim}-assisted \ac{cr}-\ac{isac} downlink framework. 
We integrated the \ac{sb} with a multi-layer \ac{sim} and adopted the \ac{bcrb} as the localization performance metric. 
We then formulated a \ac{bcrb}-minimization problem subject to an average \ac{pu} \ac{se} constraint that limits performance degradation to at most a few percent.
To design the \ac{sim} coefficients we cast the problem as a constrained nonconvex problem, reformulated it into a tractable max-min form and developed an efficient \ac{ao}-based algorithm (inner minimization via bisection and an outer step in closed form). 
Building on the \ac{sim}/\ac{dnn} analogy, we also derived analytical gradients for backpropagation and proposed an Adam-based, mini-batch learning scheme that uses a beampattern-matching loss to realize the optimal end-to-end response. 
Extensive numerical results confirm fast convergence of the proposed algorithms and demonstrate that multi-layer \acp{sim} (in particular $L=4$) can achieve near-optimal beampatterns and substantially outperform single-layer \ac{ris} designs for joint localization and interference management.
Future work will extend the present study to real-time scenarios and also account for practical implementation aspects, including hardware impairments, and develop maximum-likelihood estimators.

\vspace{-0.5cm}
\appendices
\section{Proof of Prop.~\ref{prop:problem_equivalence}}\label{app:FIM-derivatives} 
The \ac{sdr} version of the optimization problem \eqref{eq:opt_formulation_f} can be obtained as follows:
\begin{subequations}\label{eq:opt_formulation_sdr}
\begin{align}
\min_{\{\mathbf{F}[i]\}_{i=1}^I}\;&
\mathrm{BCRB}(\{\mathbf{F}[i]\}_{i=1}^I)
\\
&
\text{s.t.}\;
\epsilon[i]
\ge
\mathbf{R}_\text{pu}
\bullet
\mathbf{F}[i]
\label{const:pu_rate_sdr}
,\\
& 
\mathrm{tr}(\mathbf{F}[i])\le \delta[i],\ 
i=1,\dots ,I,
\label{const:F_pow_sdr}
\end{align}
\end{subequations}
where we have dropped the rank-1 constraints $\mathbf{F}[i]=\mathbf{f}[i]\mathbf{f}[i]^\mathsf{H}$ to make the problem convex. 
Note that:
\begin{equation}
\mathrm{BCRB}
=
\mathrm{tr}\bigl([\mathbf J_{B}^{-1}]_{1:3,\,1:3}\bigr)
=
\sum_{j=1}^{3}
\mathbf{e}_j^\mathsf{T}
\mathbf J_{B}^{-1}
\mathbf{e}_j.
\end{equation}
Next, we write the epigraph form of \eqref{eq:opt_formulation_sdr} as follows:
\begin{subequations}\label{eq:opt_formulation_sdr_epi}
\begin{align}
\min_{\{\mathbf{F}[i]\}_{i=1}^I,\ \{t_j\}_{j=1}^3}\;&
\sum_{j=1}^{3}
t_j
\\
&
\text{s.t.}\;
\begin{bmatrix}
\mathbf{J}_B\left(
\{\mathbf{F}[i]\}_{i=1}^I\right) & \mathbf{e}_j\\
\mathbf{e}_j^{\mathsf{T}} & t_{j}
\end{bmatrix}
\succeq 0,
\label{const:lmi_sdr}
\\
&
\ j=1,2,3,
\notag
\\
& 
\eqref{const:pu_rate_sdr},\ 
\eqref{const:F_pow_sdr}
\notag
\end{align}
\end{subequations}
where we have used the Schur complement technique \cite{fadakar2025hybrid, Attiah2024duality} to obtain the \ac{lmi} \eqref{const:lmi_sdr} which is equivalent to $t_j\ge \mathbf{e}_j^\mathsf{T}
\mathbf J_{B}^{-1}
\mathbf{e}_j$. 
The Lagrangian of the problem with respect to \eqref{const:lmi_sdr} is defined as:
\begin{equation}\label{def:lag}
\mathcal{L}
=
\sum_{j=1}^{3}
t_j
-
\sum_{j=1}^{3}
\mathbf{D}_{j}\bullet 
\mathbf{B}_{j}
\end{equation}
where $\mathbf{B}_{j}$ is the matrix in \eqref{const:lmi_sdr}, and $\mathbf{D}_{j}\succeq 0$ denotes the dual variables with respect to the \ac{lmi} \eqref{const:lmi_sdr}. 
By partitioning $\mathbf{D}_{j}$ as
$$
\mathbf{D}_{j}
=
\begin{bmatrix}
\overline{\mathbf{D}}_j & -\mathbf{d}_j\\
-\mathbf{d}_j^{\mathsf{T}} & d_j
,
\end{bmatrix},
$$
the Lagrangian \eqref{def:lag} can be rewritten as:
\begin{equation}\label{def:lag_rep}
\mathcal{L}
=
\sum_{j=1}^{3}
t_j(1-d_j)
-
\sum_{j=1}^{3}
\overline{\mathbf{D}}_{j}\bullet 
\mathbf{J}_{B}
+
2
\sum_{j=1}^{3}
\mathbf{e}_j^\mathsf{T}
\mathbf{d}_j.
\end{equation}
The dual function is found by minimizing the Lagrangian over the primal variables, $\{\mathbf{F}[i]\}_{i=1}^{I}$ and $\{t_j\}_{j=1}^{3}$. 
By setting the derivative with respect to $t_j$ to zero, we obtain $d_j=1$, because otherwise the minimum is unbounded.
Applying the Schur complement to the constraint $\mathbf{D}_{j}\succeq 0$, we obtain the equivalent constraint $\overline{\mathbf{D}}_j\succeq \mathbf{d}_j\mathbf{d}_j^\mathsf{T}$. 
Thus, the dual problem is obtained as:
\begin{subequations}\label{def:max_min}
\begin{align}
\max_{\{\overline{\mathbf{D}}_j\}_{j=1}^{3}}\,
\min_{\{\mathbf{F}[i]\}_{i=1}^{I}}
&-
\sum_{u=1}^{3}
\overline{\mathbf{D}}_{u}
\bullet 
\mathbf{J}_{B}
+
2
\sum_{u=1}^{3}
\mathbf{e}_j^\mathsf{T}
\mathbf{d}_j,
\\
& 
\text{s.t.}\;
\eqref{const:pu_rate_sdr},\ 
\eqref{const:F_pow_sdr}
\notag
\end{align}
\end{subequations}
Hence, since strong duality holds, the duality gap is zero, meaning the optimal values of the primal and dual problems are equal. 
This allows us to interchange $\min$ and $\max$ in \eqref{def:max_min} and first find the optimal $\overline{\mathbf{D}}_{j}$ \cite{boyd2004convex}. 
Since $\overline{\mathbf{D}}_j\succeq \mathbf{d}_j\mathbf{d}_j^\mathsf{T}$ and $\mathbf{J}_B\succeq 0$, we have $\overline{\mathbf{D}}_j\bullet \mathbf{J}_B\ge \mathbf{d}_j\mathbf{d}_j^\mathsf{T}\bullet \mathbf{J}_B$, from which we conclude that at the optimum we have $\overline{\mathbf{D}}_{j}=\mathbf{d}_j\mathbf{d}_j$. 
Moreover, after substituting this result back into \eqref{def:max_min}, we obtain \eqref{opt:equivalent_v1}. 
By substituting \eqref{eq:FIM_entry} in \eqref{def:J_B}, and then \eqref{def:J_B} in \eqref{obj:equivalent_v1}, \eqref{opt:equivalent_v1} can be represented as \eqref{opt:equivalent_v2}. 
By \cite[Prop.~3.5]{Huang2010Rank}, for any fixed vectors $\{\mathbf{d}_j\}_{j=1}^{3}$ problem \eqref{opt:equivalent_v2} admits rank-one optimal solutions of the form $\hat{\mathbf{F}}[i]=\hat{\mathbf{f}}[i]\hat{\mathbf{f}}[i]^{\mathsf{H}}$, which completes the proof.

\vspace{-0.2cm}
\section{Proof of Prop.~\ref{prop:inner_opt}}\label{app:inner_opt}
We rewrite the inner optimization problem as follows: 
\begin{align} \label{eq:inner} &\max_{\mathbf{f}\in\mathbb{C}^{N}}\quad \mathbf{f}^{\mathsf{H}} 
\mathbf{A} 
\mathbf{f} \ \nonumber &\text{s.t.}\quad \mathbf{f}^{\mathsf{H}} 
\mathbf{R} 
\mathbf{f} 
\le 
\epsilon,\ \nonumber &\qquad\qquad\lVert\mathbf{f}\rVert^2 \le \delta.
\end{align} 
%
We form the Lagrangian with multipliers $\mu\ge0$ (for the \ac{pu} constraint) and $\lambda\ge0$ (for the power constraint): 
\begin{equation} \mathcal L(\mathbf f,\mu,\lambda)=\mathbf f^{\mathsf{H}}\mathbf A\mathbf f - \mu(\mathbf f^{\mathsf{H}}\mathbf R\mathbf f - \epsilon) - \lambda(\mathbf f^{\mathsf{H}}\mathbf f - \delta). 
\end{equation} 
After setting the derivative with respect to $\mathbf{f}$ to zero we obtain: \begin{equation}\label{eq:stationarity} (\mathbf{A} - 
\mu
\mathbf{R} - 
\lambda\mathbf{I})\hat{\mathbf{f}}
= 
\mathbf{0}. 
\end{equation} 
From complementary slackness and feasibility: 
\begin{align} 
\mu\big(\hat{\mathbf{f}}^{\mathsf{H}}\mathbf{R}\hat{\mathbf{f}} - \epsilon\big) 
= 0,\ 
\lambda\big(\hat{\mathbf{f}}^{\mathsf{H}}
\hat{\mathbf{f}} 
- 
\delta
\big) 
&
= 
0,
\end{align}
Next, We consider three cases:
\vspace{-0.3cm}
\subsection{Case 1: $\lambda_{\max}(\mathbf{A})\le 0$} 
In this case it is easy to observe that $\hat{\mathbf{f}}=\mathbf{0}$.
\vspace{-0.3cm}
\subsection{Case 2: $\lambda_{\max}(\mathbf A)\ge 0$ with feasible principal eigenvector} 
Let $\mathbf{u}$ be a principal unit eigenvector of $\mathbf A$ associated to $\lambda_{\max}(\mathbf A)$. 
If 
$
\mathbf{u}^{\mathsf{H}} 
\mathbf{R} 
\mathbf{u} 
\le 
\frac{\epsilon}{\delta}, 
$
then $\hat{\mathbf{f}}=\sqrt{\delta}\,\mathbf u$ is feasible and attains the objective value $\delta\,\lambda_{\max}(\mathbf A)$, and the \ac{pu} multiplier $\mu$ is zero.
\vspace{-0.3cm}
\subsection{Case 3: $\lambda_{\max}(\mathbf A)\ge 0$ with infeasible principal eigenvector} 
If the inequality in the previous case is violated i.e., $
\mathbf{u}^{\mathsf{H}} 
\mathbf{R} 
\mathbf{u} 
>
\frac{\epsilon}{\delta}
$, the \ac{pu} interference constraint must be active at optimum (i.e. $\mu>0$) \cite{boyd2004convex}. 
Since $\hat{\mathbf{f}}\neq\mathbf{0}$ then \eqref{eq:stationarity} implies $\hat{\mathbf{f}}$ lies in the nullspace of the Hermitian matrix $\mathbf{A} - \mu\mathbf{R} - \lambda\mathbf{I}$, 
or equivalently $\hat{\mathbf{f}}$ is an eigenvector of $\mathbf{A}-\mu\mathbf{R}$ with eigenvalue $\lambda$.
Thus, if we define $\mathbf{v}=\mathbf{f}/\sqrt{\delta}$, we conclude that the optimal unit-norm direction $\mathbf{v}$ is the principal eigenvector of the Hermitian matrix: 
$
\mathbf{M}(\mu) \triangleq \mathbf A - \mu\mathbf R, 
$
and the equality constraint to enforce is:
\begin{equation}\label{eq:root_cond}
\mathbf{v}(\mu)^{\mathsf{H}} 
\mathbf{R} 
\mathbf{v}(\mu) 
= 
\frac{\epsilon}{\delta}.
\end{equation} 
Thus the problem reduces to finding a scalar $\mu\ge0$ such that the principal eigenvector $\mathbf{v}(\mu)$ of $\mathbf{M}(\mu)$ satisfies \eqref{eq:root_cond}. 
Once $\mu$ is found, the optimizer is $\hat{\mathbf{f}} = \sqrt{\delta}\,\mathbf v(\mu)$.

We propose to use bisection root-finding algorithm to the scalar function
$
h(\mu)\;=\;g(\mu)-\epsilon/\delta.
$
Under the usual non-degeneracy assumption that the principal eigenvalue of $\mathbf M(\mu)$ is simple, $\mathbf v(\mu)$ (and hence $g(\mu)$) depends continuously on $\mu$, so $h(\mu)$ is continuous.  We therefore first obtain a bracket $[\mu_\ell,\mu_u]$ (using $\mu_\ell=0$ and $\mu_u=1$ as initial values) by doubling $\mu_u\leftarrow 2\mu_u$ until $h(\mu_u)\le0$ (i.e. $g(\mu_u)\le\epsilon/\delta$), which ensures $h(\mu_\ell)\ge0$ and $h(\mu_u)\le0$.  Bisection method then iterates by evaluating $g(\mu)$ at $\mu=\frac{\mu_\ell+\mu_u}{2}$,
and replaces the left or right endpoint according to the sign of $h(\mu)$, until $|h(\mu)|<\xi_{\rm tol}$.
The details are also presented in Algorithm~\ref{alg:F_opt}.

\vspace{-0.2cm}
\section{Proof of Prop.~\ref{prop:outer}}\label{app:outer_opt}
The objective in \eqref{eq:outer_problem_rev} is the sum of three decoupled quadratic functions
$
\ell_j(\mathbf d_j)=2\mathbf e_j^\mathsf{T}\mathbf d_j-\mathbf d_j^\mathsf{T}\mathbf J_B\mathbf d_j.
$
Setting the derivative to zero, we obtain the linear normal equation
$
2\mathbf{e}_j - 2\mathbf J_B\mathbf d_j=\mathbf{0}$ which can be rewritten as $\mathbf J_B\mathbf d_j=\mathbf e_j.
$
Since $\mathbf{J}_B$ is positive definite, the linear system has the unique solution \eqref{eq:d_star_5}. 

\vspace{-0.5cm}
\section{Proof of Lemma~\ref{lemma:trace}}\label{app:lemma_trace}
First we partition $\mathbf{J}_B$ as 
\[
\mathbf{J}_B=
\begin{bmatrix}
\mathbf{U} & \mathbf{V} \\
\mathbf{V}^\mathsf{T} & \mathbf{W}
\end{bmatrix},
\]
where $\mathbf{U}\in\mathbb{R}^{3\times 3}$, $\mathbf{V}\in\mathbb{R}^{3\times 2}$ and $\mathbf{W}\in\mathbb{R}^{2\times 2}$.
Since $\mathbf{J}_B\succ 0$, the block $\mathbf{W}$ is invertible and the Schur complement of $\mathbf{W}$ is defined as
$
\mathbf{S}=\mathbf{U}-\mathbf{V}\mathbf{W}^{-1}\mathbf{V}^\mathsf{T}.
$
By standard block-inverse formula the top-left $3\times3$ block of $\mathbf{J}_B^{-1}$ equals $\mathbf{S}^{-1}$; equivalently,
\begin{equation}\label{eq:topblock}
[\mathbf{J}_B^{-1}]_{1:3,1:3}=\mathbf{S}^{-1}
\rightarrow
\operatorname{tr}\!\big([\mathbf{J}_B^{-1}]_{1:3,1:3}\big)=\operatorname{tr}(\mathbf{S}^{-1}).
\end{equation}
Let the eigenvalues of $\mathbf{S}$ be $\mu_1,\mu_2,\mu_3>0$. Then
$\operatorname{tr}(\mathbf{S})=\sum_{i=1}^3\mu_i$ and
$\operatorname{tr}(\mathbf{S}^{-1})=\sum_{i=1}^3\mu_i^{-1}$. By the Cauchy-Schwarz inequality:
\[
\Big(\sum_{i=1}^3\mu_i\Big)\Big(\sum_{i=1}^3\frac{1}{\mu_i}\Big)
\ge \Big(\sum_{i=1}^3 1\Big)^2 = 9,
\]
which gives
\begin{equation}\label{eq:csS}
\operatorname{tr}(\mathbf{S}^{-1}) \ge \frac{9}{\operatorname{tr}(\mathbf{S})}.
\end{equation}
Because $\mathbf{V}\mathbf{W}^{-1}\mathbf{V}^T\succeq 0$ we have $\mathbf{S}\preceq \mathbf{U}$ in the \ac{psd} order, and therefore
\begin{equation}\label{eq:trace-order-U}
\operatorname{tr}(\mathbf{S}) \le \operatorname{tr}(\mathbf{U}).
\end{equation}
Combining \eqref{eq:topblock}, \eqref{eq:csS} and \eqref{eq:trace-order-U} yields
\[
\operatorname{tr}\!\big([\mathbf{J}_B^{-1}]_{1:3,1:3}\big)
= \operatorname{tr}(\mathbf{S}^{-1})
\ge \frac{9}{\operatorname{tr}(\mathbf{S})}
\ge \frac{9}{\operatorname{tr}(\mathbf{U})},
\]
which proves the claim. 

\vspace{-0.1cm}
\bibliographystyle{IEEEtran}
\bibliography{Bib}

@IEEEtranBSTCTL{IEEEexample:BSTcontrol,
CTLuse_forced_etal       = "yes",
CTLmax_names_forced_etal = "1",
CTLnames_show_etal       = "1" }

@book{molisch2026wireless,
  title={Wireless communications: from fundamentals to beyond 5G},
  author={Molisch, Andreas F},
  year={2022},
  publisher={John Wiley \& Sons}
}

@book{golub2013matrix,
  title={Matrix computations},
  author={Golub, Gene H and Van Loan, Charles F},
  year={2013},
  publisher={JHU press}
}

@ARTICLE{Huang2010Rank,
  author={Huang, Yongwei and Palomar, Daniel P.},
  journal={IEEE Transactions on Signal Processing}, 
  title={Rank-Constrained Separable Semidefinite Programming With Applications to Optimal Beamforming}, 
  year={2010},
  volume={58},
  number={2},
  pages={664-678},
  doi={10.1109/TSP.2009.2031732}}

@ARTICLE{Nuria2024ISAC,
  author={González-Prelcic, Nuria and Furkan Keskin, Musa and Kaltiokallio, Ossi and Valkama, Mikko and Dardari, Davide and Shen, Xiao and Shen, Yuan and Bayraktar, Murat and Wymeersch, Henk},
  journal={Proceedings of the IEEE}, 
  title={The Integrated Sensing and Communication Revolution for {6G}: Vision, Techniques, and Applications}, 
  year={2024},
  volume={112},
  number={7},
  pages={676-723},
  doi={10.1109/JPROC.2024.3397609}}

@ARTICLE{tataria20216,
  author={Tataria, Harsh and Shafi, Mansoor and Molisch, Andreas F. and Dohler, Mischa and Sjöland, Henrik and Tufvesson, Fredrik},
  journal={Proceedings of the IEEE}, 
  title={{6G} Wireless Systems: Vision, Requirements, Challenges, Insights, and Opportunities}, 
  year={2021},
  volume={109},
  number={7},
  pages={1166-1199},
  doi={10.1109/JPROC.2021.3061701}}

@article{fadakar2025hybrid,
  title={Hybrid Codebook Design for Localization Using Electromagnetically Reconfigurable Fluid Antenna System},
  author={Fadakar, Alireza and Zhang, Yuchen and Chen, Hui and Keskin, Musa Furkan and Wymeersch, Henk and Molisch, Andreas F},
  journal={arXiv preprint arXiv:2508.21351},
  year={2025}
}

@article{boyd2004convex,
  title={Convex optimization},
  author={Boyd, Stephen},
  journal={Cambridge UP},
  year={2004}
}

@article{fascista2022ris,
  author={Fascista, Alessio and Keskin, Musa Furkan and Coluccia, Angelo and Wymeersch, Henk and Seco-Granados, Gonzalo},
  journal={IEEE Journal of Selected Topics in Signal Processing}, 
  title={{RIS}-Aided Joint Localization and Synchronization With a Single-Antenna Receiver: Beamforming Design and Low-Complexity Estimation}, 
  year={2022},
  volume={16},
  number={5},
  pages={1141-1156},
  doi={10.1109/JSTSP.2022.3177925}}

@ARTICLE{fadakar2024multi,
  author={Fadakar, Alireza and Sabbaghian, Maryam and Wymeersch, Henk},
  journal={IEEE Open Journal of the Communications Society}, 
  title={Multi-{RIS}-Assisted {3D} Localization and Synchronization via Deep Learning}, 
  year={2024},
  volume={5},
  number={},
  pages={3299-3314},
  doi={10.1109/OJCOMS.2024.3399605}}

@ARTICLE{Chen2024Multi,
  author={Chen, Hui and Zheng, Pinjun and Keskin, Musa Furkan and Al-Naffouri, Tareq and Wymeersch, Henk},
  journal={IEEE Transactions on Wireless Communications}, 
  title={Multi-{RIS}-Enabled {3D} Sidelink Positioning}, 
  year={2024},
  volume={23},
  number={8},
  pages={8700-8716},
  doi={10.1109/TWC.2024.3353387}}

@ARTICLE{An2024Codebook,
  author={An, Jiancheng and Xu, Chao and Wu, Qingqing and Ng, Derrick Wing Kwan and Di Renzo, Marco and Yuen, Chau and Hanzo, Lajos},
  journal={IEEE Wireless Communications}, 
  title={Codebook-Based Solutions for Reconfigurable Intelligent Surfaces and Their Open Challenges}, 
  year={2024},
  volume={31},
  number={2},
  pages={134-141},
  doi={10.1109/MWC.010.2200312}}

@ARTICLE{An2024SIM_DOA,
  author={An, Jiancheng and Yuen, Chau and Guan, Yong Liang and Renzo, Marco Di and Debbah, Mérouane and Poor, H. Vincent and Hanzo, Lajos},
  journal={IEEE Journal on Selected Areas in Communications}, 
  title={Two-Dimensional Direction-of-Arrival Estimation Using Stacked Intelligent Metasurfaces}, 
  year={2024},
  volume={42},
  number={10},
  pages={2786-2802},
  doi={10.1109/JSAC.2024.3414613}}

@ARTICLE{Xu2025CR,
  author={Xu, Yongqing and Li, Yong and Quek, Tony Q. S.},
  journal={IEEE Journal on Selected Areas in Communications}, 
  title={{RIS}-Enhanced Cognitive Integrated Sensing and Communication: Joint Beamforming and Spectrum Sensing}, 
  year={2025},
  volume={43},
  number={3},
  pages={795-810},
  doi={10.1109/JSAC.2025.3531531}}

@ARTICLE{You2021Wireless,
  author={You, Changsheng and Zheng, Beixiong and Zhang, Rui},
  journal={IEEE Wireless Communications Letters}, 
  title={Wireless Communication via Double {IRS}: Channel Estimation and Passive Beamforming Designs}, 
  year={2021},
  volume={10},
  number={2},
  pages={431-435},
  doi={10.1109/LWC.2020.3034388}}

@INPROCEEDINGS{Chen2024fluid,
  author={Chen, Jiangong and Xiao, Yue and Zhu, Jing and Peng, Zhendong and Lei, Xia and Xiao, Pei},
  booktitle={2024 IEEE International Conference on Communications Workshops (ICC Workshops)}, 
  title={Low-Complexity Beamforming Design for {RIS}-Assisted Fluid Antenna Systems}, 
  year={2024},
  volume={},
  number={},
  pages={1377-1382},
  doi={10.1109/ICCWorkshops59551.2024.10615874}}

@ARTICLE{fadakar2025mutual,
  author={Fadakar, Alireza and Keskin, Musa Furkan and Chen, Hui and Wymeersch, Henk},
  journal={IEEE Transactions on Cognitive Communications and Networking}, 
  title={Mutual Coupling-Aware Localization for {RIS}-Assisted {ISAC} Systems}, 
  year={2025},
  volume={11},
  number={5},
  pages={2938-2954},
  doi={10.1109/TCCN.2025.3565541}}

@INPROCEEDINGS{fadakar2025near,
  author={Fadakar, Alireza and Keskin, Musa Furkan and Chen, Hui and Wymeersch, Henk and Molisch, Andreas F.},
  booktitle={2025 IEEE International Conference on Communications Workshops (ICC Workshops)}, 
  title={Near-Field {RIS}-Assisted Localization Under Mutual Coupling}, 
  year={2025},
  volume={},
  number={},
  pages={1470-1475},
  doi={10.1109/ICCWorkshops67674.2025.11162429}}

@INPROCEEDINGS{Attiah2024duality,
  author={Attiah, Kareem M. and Yu, Wei},
  booktitle={2024 IEEE International Symposium on Information Theory (ISIT)}, 
  title={Beamforming Design for Integrated Sensing and Communications Using Uplink-Downlink Duality}, 
  year={2024},
  volume={},
  number={},
  pages={2808-2813},
  doi={10.1109/ISIT57864.2024.10619361}}

@ARTICLE{An2025sim_downlink,
  author={An, Jiancheng and Di Renzo, Marco and Debbah, Mérouane and Vincent Poor, H. and Yuen, Chau},
  journal={IEEE Transactions on Wireless Communications}, 
  title={Stacked Intelligent Metasurfaces for Multiuser Downlink Beamforming in the Wave Domain}, 
  year={2025},
  volume={24},
  number={7},
  pages={5525-5538},
  doi={10.1109/TWC.2025.3547779}}

@article{li2025stacked_ofdm,
  title={Stacked intelligent metasurfaces-enhanced {MIMO} {OFDM} wideband communication systems},
  author={Li, Zheao and An, Jiancheng and Yuen, Chau},
  journal={arXiv preprint arXiv:2503.00368},
  year={2025}
}

@ARTICLE{Liu2021RISs,
  author={Liu, Yuanwei and Liu, Xiao and Mu, Xidong and Hou, Tianwei and Xu, Jiaqi and Di Renzo, Marco and Al-Dhahir, Naofal},
  journal={IEEE Communications Surveys \& Tutorials}, 
  title={Reconfigurable Intelligent Surfaces: Principles and Opportunities}, 
  year={2021},
  volume={23},
  number={3},
  pages={1546-1577},
  doi={10.1109/COMST.2021.3077737}}

@INPROCEEDINGS{An2023SIM_muser,
  author={An, Jiancheng and Di Renzo, Marco and Debbah, Mérouane and Yuen, Chau},
  booktitle={ICC 2023 - IEEE International Conference on Communications}, 
  title={Stacked Intelligent Metasurfaces for Multiuser Beamforming in the Wave Domain}, 
  year={2023},
  volume={},
  number={},
  pages={2834-2839},
  doi={10.1109/ICC45041.2023.10279173}}

@ARTICLE{An2023SIM_Holog_MIMO,
  author={An, Jiancheng and Xu, Chao and Ng, Derrick Wing Kwan and Alexandropoulos, George C. and Huang, Chongwen and Yuen, Chau and Hanzo, Lajos},
  journal={IEEE Journal on Selected Areas in Communications}, 
  title={Stacked Intelligent Metasurfaces for Efficient Holographic {MIMO} Communications in {6G}}, 
  year={2023},
  volume={41},
  number={8},
  pages={2380-2396},
  doi={10.1109/JSAC.2023.3288261}}

@ARTICLE{Liu2025_MISO_SIM,
  author={Liu, Hao and An, Jiancheng and Alexandropoulos, George C. and Ng, Derrick Wing Kwan and Yuen, Chau and Gan, Lu},
  journal={IEEE Transactions on Cognitive Communications and Networking}, 
  title={Multi-User {MISO} with Stacked Intelligent Metasurfaces: A {DRL}-Based Sum-Rate Optimization Approach}, 
  year={2025},
  volume={},
  number={},
  pages={1-1},
  doi={10.1109/TCCN.2025.3558008}}

@ARTICLE{Shi2025joint,
  author={Shi, Enyu and Zhang, Jiayi and An, Jiancheng and Zhang, Guangyang and Liu, Ziheng and Yuen, Chau and Ai, Bo},
  journal={IEEE Transactions on Wireless Communications}, 
  title={Joint {AP-UE} Association and Precoding for {SIM}-Aided Cell-Free Massive {MIMO} Systems}, 
  year={2025},
  volume={24},
  number={6},
  pages={5352-5367},
  doi={10.1109/TWC.2025.3546927}}

@ARTICLE{Li2025SIM_NF,
  author={Li, Qingchao and El-Hajjar, Mohammed and Xu, Chao and An, Jiancheng and Yuen, Chau and Hanzo, Lajos},
  journal={IEEE Transactions on Communications}, 
  title={Stacked Intelligent Metasurface-Based Transceiver Design for Near-Field Wideband Systems}, 
  year={2025},
  volume={73},
  number={9},
  pages={8125-8139},
  doi={10.1109/TCOMM.2025.3544929}}

@ARTICLE{Niu2024SIM,
  author={Niu, Haoxian and An, Jiancheng and Papazafeiropoulos, Anastasios and Gan, Lu and Chatzinotas, Symeon and Debbah, Mérouane},
  journal={IEEE Wireless Communications Letters}, 
  title={Stacked Intelligent Metasurfaces for Integrated Sensing and Communications}, 
  year={2024},
  volume={13},
  number={10},
  pages={2807-2811},
  doi={10.1109/LWC.2024.3447272}}

@article{liu2022programmable,
  title={A programmable diffractive deep neural network based on a digital-coding metasurface array},
  author={Liu, Che and Ma, Qian and Luo, Zhang Jie and Hong, Qiao Ru and Xiao, Qiang and Zhang, Hao Chi and Miao, Long and Yu, Wen Ming and Cheng, Qiang and Li, Lianlin and others},
  journal={Nature Electronics},
  volume={5},
  number={2},
  pages={113--122},
  year={2022},
  publisher={Nature Publishing Group UK London}
}

@article{liu2023full,
  title={Full-range amplitude--phase metacells for sidelobe suppression of metalens antenna using prior-knowledge-guided deep-learning-enabled synthesis},
  author={Liu, Peiqin and Chen, Zhi Ning},
  journal={IEEE Transactions on Antennas and Propagation},
  volume={71},
  number={6},
  pages={5036--5045},
  year={2023},
  publisher={IEEE}
}

@ARTICLE{Heath2016Overview,
  author={Heath, Robert W. and González-Prelcic, Nuria and Rangan, Sundeep and Roh, Wonil and Sayeed, Akbar M.},
  journal={IEEE Journal of Selected Topics in Signal Processing}, 
  title={An Overview of Signal Processing Techniques for Millimeter Wave {MIMO} Systems}, 
  year={2016},
  volume={10},
  number={3},
  pages={436-453},
  doi={10.1109/JSTSP.2016.2523924}}

@INPROCEEDINGS{Wang2024SIM_ISAC,
  author={Wang, Ziqing and Liu, Hongzheng and Zhang, Jianan and Xiong, Rujing and Wan, Kai and Qian, Xuewen and Di Renzo, Marco and Qiu, Robert Caiming},
  booktitle={GLOBECOM 2024 - 2024 IEEE Global Communications Conference}, 
  title={Multi-user {ISAC} through Stacked Intelligent Metasurfaces: New Algorithms and Experiments}, 
  year={2024},
  volume={},
  number={},
  pages={4442-4447},
  doi={10.1109/GLOBECOM52923.2024.10901440}}

@ARTICLE{Li2025ISAC_SIM,
  author={Li, Shunyu and Zhang, Fan and Mao, Tianqi and Na, Rui and Wang, Zhaocheng and Karagiannidis, George K.},
  journal={IEEE Transactions on Vehicular Technology}, 
  title={Transmit Beamforming Design for {ISAC} With Stacked Intelligent Metasurfaces}, 
  year={2025},
  volume={74},
  number={4},
  pages={6767-6772},
  doi={10.1109/TVT.2024.3517709}}

@ARTICLE{Pan2022RIS_overview,
  author={Pan, Cunhua and Zhou, Gui and Zhi, Kangda and Hong, Sheng and Wu, Tuo and Pan, Yijin and Ren, Hong and Renzo, Marco Di and Lee Swindlehurst, A. and Zhang, Rui and Zhang, Angela Yingjun},
  journal={IEEE Journal of Selected Topics in Signal Processing}, 
  title={An Overview of Signal Processing Techniques for {RIS/IRS}-Aided Wireless Systems}, 
  year={2022},
  volume={16},
  number={5},
  pages={883-917},
  doi={10.1109/JSTSP.2022.3195671}}

@ARTICLE{Tanab2017Resource,
  author={El Tanab, Manal and Hamouda, Walaa},
  journal={IEEE Communications Surveys \& Tutorials}, 
  title={Resource Allocation for Underlay Cognitive Radio Networks: A Survey}, 
  year={2017},
  volume={19},
  number={2},
  pages={1249-1276},
  doi={10.1109/COMST.2016.2631079}}

@ARTICLE{Yang2019_6G,
  author={Yang, Ping and Xiao, Yue and Xiao, Ming and Li, Shaoqian},
  journal={IEEE Network}, 
  title={{6G} Wireless Communications: Vision and Potential Techniques}, 
  year={2019},
  volume={33},
  number={4},
  pages={70-75},
  doi={10.1109/MNET.2019.1800418}}

@ARTICLE{Liang2011CR,
  author={Liang, Ying-Chang and Chen, Kwang-Cheng and Li, Geoffrey Ye and Mahonen, Petri},
  journal={IEEE Transactions on Vehicular Technology}, 
  title={Cognitive radio networking and communications: an overview}, 
  year={2011},
  volume={60},
  number={7},
  pages={3386-3407},
  doi={10.1109/TVT.2011.2158673}}

@ARTICLE{Wang2011CR_Advances,
  author={Wang, Beibei and Liu, K.J. Ray},
  journal={IEEE Journal of Selected Topics in Signal Processing}, 
  title={Advances in cognitive radio networks: A survey}, 
  year={2011},
  volume={5},
  number={1},
  pages={5-23},
  doi={10.1109/JSTSP.2010.2093210}}

@ARTICLE{Yuan2025TTR_CR,
  author={Yuan, Jie and Zhou, Hu and Liang, Ying-Chang},
  journal={IEEE Transactions on Cognitive Communications and Networking}, 
  title={A Two-Timescale Resource Allocation Scheme for {RIS}-Aided Cognitive Radio Systems}, 
  year={2025},
  volume={},
  number={},
  pages={1-1},
  doi={10.1109/TCCN.2025.3553281}}

@ARTICLE{Dong2025joint_secure,
  author={Dong, Limeng and Huo, Yiran and Yan, Wanyu and Tang, Xiao and Li, Yong and Cheng, Wei},
  journal={IEEE Transactions on Cognitive Communications and Networking}, 
  title={Joint Secure Transmission Enhancement of Primary and Secondary Users in {RIS} Aided Spectrum Sharing Cognitive Radio Networks}, 
  year={2025},
  volume={},
  number={},
  pages={1-1},
  doi={10.1109/TCCN.2025.3541219}}

@ARTICLE{Yuan2021IRS_CR,
  author={Yuan, Jie and Liang, Ying-Chang and Joung, Jingon and Feng, Gang and Larsson, Erik G.},
  journal={IEEE Transactions on Communications}, 
  title={Intelligent Reflecting Surface-Assisted Cognitive Radio System}, 
  year={2021},
  volume={69},
  number={1},
  pages={675-687},
  doi={10.1109/TCOMM.2020.3033006}}

@ARTICLE{Zhong2022DRL_IRS_CR,
  author={Zhong, Canwei and Cui, Miao and Zhang, Guangchi and Wu, Qingqing and Guan, Xinrong and Chu, Xiaoli and Poor, H. Vincent},
  journal={IEEE Transactions on Communications}, 
  title={Deep Reinforcement Learning-Based Optimization for {IRS}-Assisted Cognitive Radio Systems}, 
  year={2022},
  volume={70},
  number={6},
  pages={3849-3864},
  doi={10.1109/TCOMM.2022.3171837}}

@ARTICLE{Ge2024RIS_CSS_CR,
  author={Ge, Jungang and Liang, Ying-Chang and Wang, Shuo and Sun, Chen},
  journal={IEEE Transactions on Wireless Communications}, 
  title={{RIS}-Assisted Cooperative Spectrum Sensing for Cognitive Radio Networks}, 
  year={2024},
  volume={23},
  number={9},
  pages={12547-12562},
  doi={10.1109/TWC.2024.3393516}}

@ARTICLE{Hui2024RIS_CR_HWI,
  author={Hui, Hao and Zou, Yulong and Li, Yizhi and Zhai, Liangsen and Ning, Boyu},
  journal={IEEE Transactions on Vehicular Technology}, 
  title={Robust Beamforming Design for {RIS}-Assisted Cognitive Radio Systems With Hardware Impairments}, 
  year={2024},
  volume={73},
  number={12},
  pages={19080-19095},
  doi={10.1109/TVT.2024.3443131}}

@ARTICLE{Zhao2023RIS_CR,
  author={Zhao, Bai and Lin, Min and Xiao, Shengjie and Cheng, Ming and Zhu, Wei-Ping and Al-Dhahir, Naofal},
  journal={IEEE Transactions on Vehicular Technology}, 
  title={Robust Beamforming for {RIS} Enhanced Transmissions in Cognitive Radio Networks}, 
  year={2023},
  volume={72},
  number={5},
  pages={6800-6804},
  doi={10.1109/TVT.2022.3229492}}

@ARTICLE{Wu2022secure_RIS_CR,
  author={Wu, Xuewen and Ma, Jingxiao and Gu, Chenwei and Xue, Xiaoping and Zeng, Xin},
  journal={IEEE Transactions on Vehicular Technology}, 
  title={Robust Secure Transmission Design for {IRS}-Assisted mmWave Cognitive Radio Networks}, 
  year={2022},
  volume={71},
  number={8},
  pages={8441-8456},
  doi={10.1109/TVT.2022.3172293}}
\end{document}